%% file: main-submission.tex
\documentclass[12pt]{article}

\usepackage{preamble}

\newif\ifidentified
\identifiedtrue %
\newcommand{\pp}{d} %

\usepackage{multirow}
\usepackage{makecell}

\title{Locally Private Histograms in All Privacy Regimes}
\ifidentified
\author{Cl\'ement L. Canonne\thanks{University of Sydney. Email: \email{clement.canonne@sydney.edu.au}.} \and Abigail Gentle\thanks{University of Sydney. Email: \email{agen2864@uni.sydney.edu.au}\newline This work was done in part while the authors were visiting the Simons Institute for the Theory of Computing.
}}
\fi

\begin{document}

\maketitle

\begin{abstract}
Frequency estimation, a.k.a. histograms, is a workhorse of data analysis, and as such has been thoroughly studied under differentially privacy. In particular, computing histograms in the \emph{local} model of privacy has been the focus of a fruitful recent line of work, and various algorithms have been proposed, achieving the order-optimal $\ell_\infty$ error in the high-privacy (small $\varepsilon$) regime while balancing other considerations such as time- and communication-efficiency. However, to the best of our knowledge, the picture is much less clear when it comes to the medium- or low-privacy regime (large $\varepsilon$), despite its increased relevance in practice. In this paper, we investigate locally private histograms, and the very related distribution learning task, in this medium-to-low privacy regime, and establish near-tight (and somewhat unexpected) bounds on the $\ell_\infty$ error achievable. As a direct corollary of our results, we obtain a protocol for histograms in the \emph{shuffle} model of differential privacy, with accuracy matching previous algorithms but significantly better message and communication complexity. 

Our theoretical findings emerge from a novel analysis, which appears to improve bounds across the board for the locally private histogram problem. We back our theoretical findings by an empirical comparison of existing algorithms in all privacy regimes, to assess their typical performance and behaviour beyond the worst-case setting.
\end{abstract}

\section{Introduction}
Frequency estimation is a fundamental problem of statistics: besides its use for basic surveying, it is also used as a building block in distribution learning, identifying heavy hitters in sparse domains, and regression, correlation or covariance analysis. As such, frequency estimation (and the closely related problem of heavy hitters) have been thoroughly studied in a variety of settings, ranging from streaming to privacy-preserving analytics. We here focus on the latter, and specifically on the \emph{local} model of differential privacy (LDP), where the data is distributed across a large number of users, and each datum is subject to stringent differential privacy requirements.
This question has, over the past decade, received a lot of attention, starting with~\cite{bassily2015,HsuKR12}; as we elaborate in~\cref{ssec:prior}, much is known about locally private frequency estimation, to the point that it may seem the question has been, already fully resolved~--~both in theory and practice. However, most recent large-scale implementations of local privacy (see, \eg~\cite{Apple17,CormodeJKLSW18}) must balance many efficiency objectives, including the bandwidth and computational requirements (as, among others, a proxy for energy consumption), and, above all, the estimation accuracy. Keeping this accuracy in check has, in turn, led to the common use of relatively large values of the ``privacy parameter'' $\priv$, far from the original rule of thumb of $\priv \ll 1$.\footnote{See, for instance, the list of real-case use of differential privacy maintained here, which includes use cases of LDP along with the correspond stated values of $\priv$: \url{https://desfontain.es/blog/real-world-differential-privacy.html}.}

In view of the relevance of frequency estimation to privacy-preserving algorithms and deployments, it is crucial that theory be informed by, and applicable to practice; 
and that constant factors, bounds on achievability, and the behaviour of estimation error in all parameter regimes are well understood. Yet, to the best of our knowledge, most if not all of the previous theoretical work on locally private frequency estimation focuses on the high-privacy $\priv \ll 1$ regime, leaving the low-privacy regime by and large uncharted. Addressing this gap in our theoretical understanding is the focus of our work.

\paragraph{The question and setting.} We will focus on examining the standard worst-case estimation error, or the $\lp[\infty]$-error as it is most relevant to the problem, with special importance for controlling error rates in heavy hitters. Specifically, when estimating the empirical frequencies $q=(q_1,\dots,q_\ab)$ over a universe $\cX$ of size $\ab$ from $\ns$ users, the (expected) error is given by
\begin{equation}
    \bEE{\norminf{\hat{q}-q}} = \bEE{\max_{1\leq i\leq \ab}\abs{\hat{q}_i-q_i}}
\end{equation}
where the expectation is taken over the (possibly randomised) algorithm given as input the users' data, and $\hat{q}$ is the vector of estimated frequencies. We are interested in the worst-case error over all possible datasets, that is, the supremum of the above quantity over all inputs $X\eqdef (X_1,\dots, X_\ns)\in \cX^\ns$ (equivalently, all $q$). We will seek algorithms (``protocols'') to minimize this worst-case error under the constraint of (pure) local privacy (see~\cref{sec:preliminaries}), parameterized by the privacy parameter $\priv>0$. As in most distributed settings, one can consider several variants depending on whether a common random seed is available to the users (public-coin protocols) or not (private-coin protocols). It is known that for locally private frequency estimation allowing public-coin protocols can provably reduce the communication requirements~\cite{AcharyaS:19};\footnote{Specifically, the cited paper establishes a separation between public- and private- coin protocols, showing that the former can achieve vanishing error even with constant-length communication per user, a setting where the latter must incur $\Omega(1)$ error. To complement this, it is known that using public randomness one can always reduce the per-player communication to $\clg{\priv}$ bits~\cite{bassily2015}.}   similarly, one can allow some back and forth between users and the central server (interactive protocols). However, the public-coin or interactive settings come at an increased deployment costs, and are often less easy (or even impossible) to implement as they require either broadcasting from the center or sustained two-way communication between the parties. As our focus is chiefly on analyzing the most versatile setting, we hereafter restrict ourselves for our algorithms to the private-coin (non-interactive) setting; however, we mention that our lower bounds apply to the public-coin setting as well.

A closely related question is that of \emph{distribution learning} under $\lp[\infty]$ loss, which differs from frequency estimation in that the dataset $X$ is not arbitrary, but instead assumed to be drawn \iid\ from an unknown probability distribution $q$ over $\cX$: in this sense, distribution learning is an ``easier'' problem than frequency estimation, and an algorithm for the latter implies one for the former.\footnote{There are some subtleties here, but it is worth noting, as a baseline, that absent privacy constraints it is well-known that the error of learning to $\lp[\infty]$ scales as $1/\sqrt{n}$ with no dependence on the alphabet size $\ab$; while frequency estimation can be done with zero error, as each user can simply send their data.}  Formal definitions, as well as the relation between these two questions, can be found in~\cref{sec:preliminaries}. In this paper, we therefore focus on frequency estimation, and will point out the implications for distribution learning as corollaries.

\paragraph{Connection to shuffle privacy.} To conclude this subsection, we mention that besides the use of large values of $\priv$ in practical deployments of locally private algorithms, which calls for a better understanding of the tasks in this parameter regime, another motivation for studying the ``low-privacy regime'' comes from the emergence of another model of differential privacy, \emph{shuffle privacy}~\cite{CheuSUZZ19,erlingsson2020}, and its increasingly widespread adoption. Indeed, it is known that one generic way to obtain shuffle private algorithms is via the so-called ``amplification-by-shuffling'' technique (see, \eg~\cite{FeldmanMT21}), whereby a locally private algorithm with high $\priv$ (low privacy) yields a shuffle private protocol with small $\priv$ (high privacy), and the same number of messages and communication cost. This makes characterizing the low-privacy regime of locally private algorithms a very consequential question, especially for fundamental tasks such as frequency estimation and the related distribution learning.

\subsection{Prior work}
    \label{ssec:prior}
Many innovative algorithms have been proposed in recent years for locally private histograms and distribution estimation~\cite{ErlingssonPK14,AcharyaS:19,ChenKO23,YeB17,WangHWNXYLQ16,pmlr-v139-feldman21a,feldman22a}, with a subset focusing on $\lp[\infty]$-error.\footnote{Indeed, for distribution estimation, a common error measure in the literature is $\lp[1]$, as it corresponds to total variation distance.} A lower bound of $\lp[\infty]=\Omega(\sqrt{\log\ab/\priv^2\ns})$ for frequency estimation under local differential privacy was established by~\cite{bassily2015} for most regimes of $\ab$, $\ns$, and small privacy parameter $\priv$; and various LDP protocols asymptotically achieving this bound have been proposed (see~\cref{tab:previous-results}).

Most recently, upper bounds on $\lp[\infty]$ for frequency estimation were derived in~\cite{ChenKO23}, and evaluated empirically in~\cite{feldman22a}. For large $\priv$, \ie the ``low-privacy regime,'' to the best of our knowledge the best current results are (1) a bound of $O(\sqrt{\log\ab/\priv\ns})$ on the $\lp[\infty]$ error rate established in~\cite{ChenKO23}, and, (2)~an upper bound of $O(\sqrt{\log\ab/(e^{\priv/2}\ns)} + (\log\ab)/\ns)$ from~\cite{HuangQYC22} (Theorem~2), based on CountSketch: note that this error does not vanish as $\priv$ grows.

In~\cite{pmlr-v139-feldman21a}, it was stated that under most commonly used metrics ($\lp[\infty],\lp[2]^2,\lp[1]$), the error is primary driven by the variance of an algorithm. We show that at least for $\lp[\infty]$ a close look at sub-Gaussian, and sub-gamma behaviour can provide a stronger understanding of the error in low-privacy regimes which are of particular interest in the shuffle model~\cite{FeldmanMT21}. Additionally,~\cite{ChenKO23} raise the question of whether the upper bound of $O(\sqrt{\log\ab/\priv\ns})$, is tight in the low privacy regime. We show in~\cref{theo:optimal:rappor:ub} that, quite surprisingly, it is not.

\begin{table}[h]\footnotesize
    \centering
    \begin{tabular}{|c|c|c|c|c|c|}\hline
         Protocol & Private-coin & Communication & $\lp[\infty]$ error & $\lp[2]^2$ error \\\hline
         \multirow{3}{*}{RAPPOR} & \multirow{3}{*}{$\checkmark$} & \multirow{3}{*}{$\ab$} & $\sqrt{\frac{\log\ab}{\ns\min(1, \priv^2)}}$ & \multirow{3}{*}{$\frac{\ab e^{\priv/2} }{\ns\Paren{e^{\priv/2}-1}^2}$} \\
         &&& \cellcolor{blue!10}$\sqrt{\frac{\log\ab}{\ns\min(\priv, \priv^2)}}$ &\\
         &&& \cellcolor{blue!10}$\sqrt{\frac{\log\ab}{\ns e^{\priv/2}}} + \frac{\log\ab}{\ns\priv}\cdot \log\ns$ &\\\hline
         \makecell{Subset Selection\\\cite{WangHWNXYLQ16}} &$\checkmark$ &$\frac{\ab}{e^\priv}\max(1,\priv)$ &? & $\frac{\ab e^\priv}{\ns\Paren{e^\priv-1}^2}$ \\\hline
         \makecell{(G)HR\\\cite[Theorem~7]{AcharyaS:19}}& $\checkmark$ & $\log\ab$ & ? & $\frac{\ab e^\priv}{\ns\Paren{e^\priv-1}^2}$ \\\hline
         \makecell{RHR\\\cite[Theorem~3.1]{ChenKO23}} & $\times$ & $\numbits$& ? & $\frac{\ab}{\ns\min\Paren{\Paren{e^{\priv/2}-1}^2, e^\priv, 2^{\numbits}, \ab}}$ \\\hline
         \makecell{No name\\\cite[Theorem~3.1]{ChenKO23}} & $\times$ & $\numbits$ & $\sqrt{\frac{\log\ab}{\ns\min(\priv, \priv^2, \numbits)}}$ & $\frac{\ab e^{2\priv/\numbits} }{\ns\numbits\Paren{e^{\priv/\numbits}-1}^2}$ $(\dagger)$ \\\hline
         \makecell{CountSketch-based\\\cite{HuangQYC22}} & $\times$ & $\max(\priv,\log\frac{1}{\priv})$ & $\sqrt{\frac{\log\ab}{n(e^{\priv/2}-1)^2}}+\frac{\log\ab}{n}$& $\frac{\ab e^{\priv/2} }{\ns\Paren{e^{\priv/2}-1}^2}$ \\\hline
         \makecell{PGR\\\cite{feldman22a}} & $\checkmark$ & $\log\ab$ & \cellcolor{blue!10} $\sqrt{\frac{\log\ab}{ne^\priv}}+\frac{\log\ab}{n\priv}\cdot\log\ns$ & $\frac{\ab e^{\priv} }{\ns\Paren{e^{\priv}-1}^2}$ \\\hline
    \end{tabular}
    \caption{Selection of best known upper bounds for communication, $\lp[\infty]$ and $\lp[2]^2$-error, where the shaded cells are our results ($\dagger$ \textit{indicates results not established in the corresponding paper, but that we derive for completeness.}) As discussed above, ``private-coin'' refers to the fact that the users do not require access to a common random seed (typically easier to implement than ``public-coin'' protocols, where they do). We note that all public-coin protocols mentioned here can be made private-coin at the cost of an $O(\log\ab)$ blowup factor in the communication cost.
    }
    \label{tab:previous-results}
\end{table}
Histograms have also seen detailed analysis in the shuffle model as both a practical question~\cite{GhaziG0PV21,Ghazi0MP20,CheuSUZZ19,CheuZ22,Bittau17} and as a benchmark for reasoning about the power of the shuffle model~\cite{balcer2020}. As two points of comparison we highlight~\cite{GhaziG0PV21} which introduced a multiple--message protocol for shuffled histograms achieving expected $\lp[\infty]$-error $O({\sqrt{\log\ab\log(1/\privdelta)}/(\priv\ns)})$ with $O({\ab^{1/100}})$ rounds of communication, and~\cite{balcer2020} which demonstrated expected maximum error of $O(\log(1/\privdelta)/(\priv^2\ns))$ with $\ab+1$ rounds. While our results do not achieve the latter bound, we show surprisingly that the former can be achieved in some fairly permissive parameter regimes with only a single message.
\subsection{Overview of results}
We here summarise our main results, and briefly discuss the underlying techniques. While we focus on two protocols in particular, the techniques themselves are broadly applicable. Our first set of results concerns the $\lp[\infty]$ error rate achievable by LDP protocols for frequency estimation. We first recall, as a baseline, a general transformation which converts any LDP protocol with optimal error in the high-privacy regime into another LDP protocol with \emph{reasonable error} in the low-privacy regime as well, at the cost of a blowup in communication:
\begin{proposition}[Informal; see~\cref{prop:generic:transformation}]
Let $A$ be any locally private protocol for frequency estimation with expected $\lp[\infty]$ error $\bigO{\sqrt{{\log\ab}/{(\ns\priv^2)}}}$ for $\priv \leq 1$, using $\numbits$ bits of communication per user. Then there is a locally private protocol $A'$ achieving error
\[
\bigO{\sqrt{\frac{\log\ab}{\ns\min(\priv,\priv^2)}}}
\]
for all $\priv > 0$, using $\numbits\clg{\priv}$ bits of of communication per user.
\end{proposition}
We emphasise that this transformation is not new, and mimics an argument found in, \eg~\cite{ChenKO23}. This general transformation is quite appealing, as it provides (in theory) good performance in the low-privacy regime ``for free'' given any good enough LDP protocol for the high-privacy one. However, this comes at a price: first, the communication blowup, which could be impractical; second, a loss in constant factors, which while relatively small might still be prohibitive; and, perhaps more importantly, \emph{this requires changing the existing algorithm} (and as a result the data analysis pipeline), which is often a significant hurdle. Our second result, focusing on one of the earliest, versatile, and (at least in its ``vanilla'' version) conceptually simple LDP protocols for frequency estimation, RAPPOR~\cite{ErlingssonPK14}, shows that this transformation is not actually necessary, and that RAPPOR actually achieves this improved bound \emph{without} modification:
\begin{theorem}[Informal; see~\cref{theo:rappor:improved}]
    \label{theo:rappor:improved:informal}
The (simple version of) RAPPOR achieves expected $\lp[\infty]$ error
\[
\bigO{\sqrt{\frac{\log\ab}{\ns\min(\priv,\priv^2)}}}
\]
for all $\priv > 0$, using $\ab$ bits of of communication per user.
\end{theorem}
\noindent In comparison, using the generic transformation above on RAPPOR would require $\clg{\priv}$ $\ab$-bit messages per user, and, in terms of worst-case theoretical bounds, an expected error worse by a factor $\simeq 2.04$. %
As such, our first result can be summarized as saying that \emph{analyzing (again) an existing algorithm can be better than modifying it}~--~and, quite importantly, that it may not be necessary to change an existing algorithmic pipeline to inherit better guarantees.

The proof of the above result relies on a careful analysis of the expected maximum of sums of Bernoulli random variables, and specifically on a fine-grained analysis of their subgaussian behaviour in the ``highly biased'' regime. While RAPPOR is particularly amenable to this analysis, we believe that this technique is highly general and applicable to a broad range of LDP protocols, for example those following the general ``scheme template'' of~\cite{ASZ:18:HR}. \smallskip

Yet, trying to establish optimality of this $\min(\priv,\priv^2)$ scaling turned out to be very challenging. And indeed, \emph{this is for a good reason:} as we show, it is actually possible to achieve significantly better error rate in the low-privacy regime~--~and, surprisingly, this much better error is again attained by RAPPOR, out-of-the-box:
\begin{theorem}[Informal; see~\cref{theo:optimal:rappor:ub}]
\label{theo:rappor:optimal:informal}
The (simple version of) RAPPOR achieves expected $\lp[\infty]$ error
\[
\tildeO{\max\Paren{\sqrt{\frac{\log\ab}{\ns e^{\priv/2}}},\frac{\log\ab}{\ns\priv}}}
\]
for all $\priv \geq 1$, using $\ab$ bits of of communication per user.
\end{theorem}
\noindent The $\tilde{O}$ notation here only hides a single $\log\ns$ factor in the second term. As the $\lp[\infty]$ error is always at most $1$, up to this $\log\ns$ factor this is always better than~\cref{theo:rappor:improved:informal}.

The proof of this result requires going beyond the sub-gaussian concentration behaviour alluded to before, and instead analyse the maximum of these sums of Bernoulli random variables as \emph{sub-gamma} random variables. More precisely, we draw upon very recent work by~\cite{CohenK23} (see also~\cite{BlanchardV24}) on ``local Glivenko--Cantelli'' results, which provide refined concentration bounds for mean estimation of high-dimensional product distributions -- in our case, the distributions over $\{0,1\}^\ab$ arising from the use of RAPPOR.

The above result is appealing, in that it not only yields better error rate than previously known (or, in many cases, believed to be possible) in the $\priv \gg 1$ regime; but also in that it is achieved ``for free'' by an existing and widely used algorithm. However, it does have one negative aspect: namely, that RAPPOR is, in its standard version, very inefficient from a communication point of view, as it require one $\ab$-bit message from each user~--~far from the ideal $O(\log\ab)$ bits one could hope for. Luckily, as mentioned above, the analysis underlying~\cref{theo:rappor:optimal:informal} is quite general, and applies to a broad range of locally private estimation algorithms. While it does not lead to the improved bound for \emph{all} such protocols, we show that it applies, for instance, to the recently proposed Projective Geometry Response protocol of Feldman, Nelson, Nguyen, and Talwar~\cite{feldman22a}, whose performance was originally only analyzed for $\lp[2]$ error:
\begin{theorem}[Informal; see~\cref{theo:optimal:pgr:ub}]
\label{theo:pgr:optimal:informal}
Projective Geometry Response (PGR) achieves expected $\lp[\infty]$ error
\[
\tildeO{\max\Paren{\sqrt{\frac{\log\ab}{\ns e^{\priv}}},\frac{\log\ab}{\ns\priv}}}
\]
for all $\priv \geq 1$, using $O(\log \ab)$ bits of of communication per user.
\end{theorem}
Note that again this is achieved ``out-of-the-box'' by an existing algorithm, without any modification! As a result, this inherits all of the features of PGR its authors originally established: crucially, its computational efficiency. We also point out that the stated guarantee is even slightly better than that of RAPPOR, as the first term now features an $e^{-\priv}$ dependence (instead of the $e^{-\priv/2}$ of~\cref{theo:rappor:optimal:informal}).\smallskip

We then turn to prove optimally of this error rate, and show that this is, up to a logarithmic factor, optimal for any LDP protocol:
\begin{theorem}[Informal; see~\cref{thm:expected-max-lb}]
\label{thm:informal-lower-bound}
Any LDP protocol for frequency estimation must have, in the worst case, expected $\lp[\infty]$ error
\[
\bigOmega{ \sqrt{\frac{\log\ab}{\ns\priv^2}} }
\]
for $\priv \in (0,1]$, and
\[
\bigOmega{\max\Paren{\sqrt{\frac{\log\ab}{\ns e^\priv}},\frac{\log\ab}{\ns\priv}}}
\]
for $\priv \geq 1$.
\end{theorem}
We remark that this lower bound also applies to the ``easier'' question of distribution learning, as we will state momentarily. The first part of this lower bound, as mentioned before, was already known; the second part is new, and, while not necessarily difficult to show in hindsight, does require significant care in combining inequalities between various information-theoretic quantities to avoid ending up with a vacuous bound. In this sense, our main contribution for the lower bound is to establish it in a self-contained, streamlined fashion, by drawing on the ``chi-square contraction'' framework of~\cite{AcharyaCT:IT1}; and~--~importantly~--~that it matches the upper bound obtained earlier in \emph{both} the high- and low-privacy regimes. Another interesting aspect of this lower bound is that all three terms are derived separately, but from the same family of hard instances: a dataset where almost all users hold the same element.

\paragraph{Implications for shuffle privacy.} By combining our new results on LDP protocols for histograms in the low-privacy regime with known ``amplification-by-shuffling'' results, we are able to obtain a simple, robust shuffle (approximate) DP protocol using only \emph{one} message per user of logarithmic length, while achieving state-of-the-art:
\begin{theorem}[Informal; see~\cref{theo:pgr:shuffle}]
    Shuffling the Projective Geometry Response protocol achieves expected $\lp[\infty]$ error 
    \[
    \bigO{\frac{\sqrt{\log(\ab)\log(1/\privdelta)}}{n\priv}}
    \]
    for all $\priv\in[\Omega(1/\sqrt{\ns}),1]$, with $O(\log\ab)$ bits of communication and a single message per user.
\end{theorem}
This matches that of the best known shuffle DP protocols, but with a much lower message and communication complexity. It was, to the best of our knowledge, still open whether achieving these ``best of three worlds'' guarantees was possible.

\paragraph{LDP protocols beyond worst-case bounds.} All the above results, while motivated by practical considerations, are quite theoretical in nature; To complement the new analysis we include a second set of results focusing on understanding how well many of these LDP protocols (namely, those from a large class of protocols, ``subset-based'', which includes, \eg Subset-Selection, (Generalized) Hadamard Response, Recursive Hadamard Response) perform \emph{in practice}. The proof of \cref{theo:rappor:improved} did not immediately admit application to these aforementioned protocols which motivated an empirical investigation of these protocols' performance. From this we found that Subset Selection seems to almost match the lower bound in \cref{fig:percentiles}, and that it appears to obey the bounds derived for RAPPOR in \cref{fig:error-by-eps}. Our empirical investigations, detailed in \cref{sec:empirical-findings}, show that there is a large performance gap in practice, provide comparison with derived and existing bounds, and demonstrate that most protocols (RAPPOR being the exception) have a surprising distribution dependence in their $\lp[\infty]$--error.\bigskip

We conclude by providing the corollary of our results for distribution learning, using the known connection to frequency estimation (along with the standard lower bound of $\Omega(1/\sqrt{\ns})$ on the error rate for the question, absent privacy constraints):
\begin{corollary}
    \label{coro:theo:ub:lb}
    For distribution learning under local privacy constraints, the minmax $\lp[\infty]$ error achievable is at most
    \[
        \tildeO{\max\Paren{\sqrt{\frac{\log\ab}{\ns e^{\priv}}},\frac{\log\ab}{\ns\priv}, \frac{1}{\sqrt{\ns}}}}
    \]
    for all $\priv \geq 1$, and
    \[
        \bigO{\sqrt{\frac{\log\ab}{\ns \priv^2}}}
    \]
    for $\priv \in (0,1]$; and this is tight, up to a logarithmic factor (in $\ns$). Moreover, the upper bound is attained by PGR, using $O(\log \ab)$ bits of of communication per user. 
\end{corollary}

\paragraph{Organisation.} After providing some background and setting notation, we establish our theoretical results (upper bounds on the error rate) in~\cref{sec:upper-bounds}, followed by information-theoretic lower bounds in~\cref{subsec:lower-bounds} and applications to shuffle privacy in~\cref{sec:shuffle}. \cref{sec:empirical-findings} then contains our empirical findings; finally, we discuss our results and potential future work in~\cref{sec:discussion}. Some miscellaneous proofs, omitted for clarity of exposition, can be found in the appendix.

\section{Preliminaries and notation.}

\label{sec:preliminaries}
\input{sec-preliminaries}

\section{Better algorithms in the low-privacy regime}
\label{sec:upper-bounds}
Although the order-optimal $\lp[\infty]$ error rates for LDP frequency estimation and distribution learning are well-understood by now in the high-privacy regime, with many distinct algorithms achieving the tight
\[
\bigO{\sqrt{\frac{\log\ab}{\ns\priv^2}}}
\]
expected error bound, much less is understood about the best achievable error for large, or even medium, values of $\priv$. In this section, we first revisit (and slightly generalize) an idea from~\cite{ChenKO23}, which shows how to convert any protocol ``optimal in the high-privacy regime'' to a related protocol ``good enough in the low-privacy regime as well'', at the price of a blowup in communication (\cref{subsec:generic-transformation}). We then focus on the specific example of RAPPOR, showing that~--~perhaps surprisingly~--~this simple protocol does already achieve the same bound by \emph{without any modification} nor communication blowup (\cref{subsec:rappor-improved}). We finally show that for distribution learning, this very same RAPPOR in fact achieves an \emph{even better} error rate than this, leveraging very recent results on concentration of the empirical mean for high-dimensional distributions, due to~\cite{CohenK23,BlanchardV24} (\cref{sec:tighter-analysis}).

\subsection{A generic transformation, and a baseline}
\label{subsec:generic-transformation}
Here we prove the following generic statement:
\begin{proposition}
\label{prop:generic:transformation}
Suppose there exists a symmetric\footnote{\ie where all users use the same randomiser.} LDP protocol $A$ for frequency estimation achieving expected $\lp[\infty]$ error
\[
 \mathcal{E}(\ns,\ab,\priv) = \bigO{\sqrt{\frac{\log\ab}{\ns\min(1,\priv^2)}}}
\]
for $\priv \leq 1$, with $m$ messages and $\numbits$ bits of communication per user. Then, for every integer $L\geq 1$, there exists an LDP protocol $A'$ for frequency estimation achieving expected $\lp[\infty]$ error
\[
 \mathcal{E}\mleft(\ns\cdot\min(\clg{\priv}, L),\ab,\frac{\priv}{\min(\clg{\priv}, L)}\mright) = \bigO{\sqrt{\frac{\log\ab}{\ns\min(L,\priv,\priv^2)}}}
\]
for all $\priv > 0$, with $m\min(\clg{\priv}, L)$ messages and $\numbits\min(\clg{\priv}, L)$ bits of communication per user. 
\noindent Further, if $A$ is a public-coin (resp. private-coin) protocol, then so is $A'$.
\end{proposition}
The idea behind this result is not new, and is used (in a slightly less general form) in~\cite[Appendix~E.4]{ChenKO23}. We restate and prove it in this paper in a self-contained form, as we believe it to be of broader interest.

As a direct corollary (setting $L=\clg{\priv}$), applying the above to Hadamard Response and RAPPOR, for instance, we obtain the following bounds:
\begin{corollary}
\label{theo:hr:modified}
	For every $\priv>0$, there exists a private-coin $\priv$-LDP protocol (namely, a modification of Hadamard Response) for frequency estimation with expected $\lp[\infty]$ error
	\[
        \bigO{\sqrt{\frac{\log\ab}{\ns\min(\priv,\priv^2)}}}
    \]
    using $\clg{\priv}$ messages per user, and $\log\ab + O(1)$ bits per message.
\end{corollary}
\begin{corollary}
\label{theo:rappor:modified}
	For every $\priv>0$, there exists a private-coin $\priv$-LDP protocol (namely, a modification of RAPPOR) for frequency estimation with expected $\lp[\infty]$ error
	\[
        \bigO{\sqrt{\frac{\log\ab}{\ns\min(\priv,\priv^2)}}}
    \]
    using $\clg{\priv}$ messages per user, and $\ab$ bits per message.
\end{corollary}

\subsection{Tighter analysis for RAPPOR}
\label{subsec:rappor-improved}%

To do so, we first recall some facts and notation about RAPPOR, which will help with the analysis of its performance. 
A common simplification of the RAPPOR protocol~\cite{ASZ:18:HR, AcharyaS:19} is to parameterise it by $\priv$ as follows. 
\begin{enumerate}
	\item Given an input $x_i\in[\ab]$, one-hot encode it onto a $\ab$-bit vector $e_x$ s.t. only the $x$'th bit is 1.
	\item Flip each bit independently with probability $\frac{1}{e^{\priv/2} + 1}$ and send the resulting noisy bit-vector $Y_i$ to the server.\footnote{This is just the ``Permanent Randomised Response'' step of RAPPOR parameterised by $\priv$.}
	\item The server then receives $\ns$ noisy bit vectors, computes $\bar{Y}\gets \frac{1}{\ns}\sum\limits_{i=1}^\ns Y_i$, and estimates 
	\begin{equation}
		\hat{q}=\frac{e^{\priv/2}+1}{e^{\priv/2}-1}\bar{Y} - \frac{1}{e^{\priv/2}-1} \mathbf{1}_\ab
	\end{equation}
where $\boldsymbol{1}$ is the 1-vector of size $\ab$.
\end{enumerate}
A standard fact, which we recall here for completeness, is that $\hat{q}$ defined above is an unbiased estimator for the true vector of frequencies $q$:
\begin{lemma}[Expectation of $\hat{q}$]
\label{lemma:rappor-unbiased}
    We have $\bEE{\hat{q}} = q$.
\end{lemma}
We then have $\bEE{\norminf{\hat{q}-q}} = \frac{e^{\priv/2}+1}{e^{\priv/2}-1} \bEE{\norminf{\bar{Y}-\bEE{\bar{Y}}}}$, and so to bound the expected $\lp[\infty]$ error is suffices to bound that of the expected maximum deviation of the $\bar{Y}_j$'s from their mean. Now, each $\bar{Y}_j$ is the (normalised) sum of $\ns$ independent Bernoulli random variables: the standard way to analyse this maximum is to recall that a sum of $\ns$ independent Bernoullis is a sub-gaussian random variable with parameter\footnote{Importantly, note that this may not coincide with the variance of $X$, although it does in some important cases (e.g., for a Gaussian r.v.).} at most $\frac{\ns}{4}$. Standard results (see, for example~\cite[Chapter 2]{Wainwright_2019}) on the maximum of $\ab$ (not necessarily independent) sub-Gaussian random variables then give
\[
\bEE{\norminf{\hat{q}-q}} \leq \frac{e^{\priv/2}+1}{e^{\priv/2}-1} \sqrt{\frac{\log\ab}{2\ns}}
\]
which indeed behaves as \smash{$O(\sqrt{\log\ab/(\ns\priv^2)})$} for small $\priv$. However, the bound quickly degrades for large $\priv$, and only yields $O(\sqrt{\log\ab/\ns})$ (no dependence on the privacy parameter at all!) as $\priv$ grows.\smallskip 

To remedy this, we will need a tighter analysis of the subgaussian parameter of Bernoulli random variables in the ``very biased'' regime (which is the one we have to handle for large $\priv$, as then $\frac{1}{e^{\priv/2}+1} \approx 0$). Specifically, we will rely on the following characterisation of the sub-gaussian norm $\sigma^2(p)$ of a (centered) $\bernoulli{p}$ random variable, known as the Kearns--Saul inequality, and which can be found in, \eg \cite{buldygin2013}:
\begin{equation}
\label{eq:subg-norm}
  \sigma^2(p) = \begin{cases}
0, &p\in\{0,1\} \\
\frac{1}{4},  & p=\frac{1}{2} \\
\frac{2p-1}{2\log\frac{p}{1-p}}
\end{cases}  
\end{equation}
Importantly, this expression is symmetric: $\sigma^2(p) = \sigma^2(1-p)$ for all $p\in[0,1]$. In our setting, each $Y_j$ is the sum of $\ns$ independent Bernoulli random variables with one of two symmetric parameters, $p\eqdef \frac{1}{e^{\priv/2}+1}$ or $1-p = \frac{e^{\priv/2}}{e^{\priv/2}+1}$.  The sub-Gaussian norm at play is then 
$
    \sigma^2(p) = \frac{e^{\priv/2} - 1}{(e^{\priv/2} + 1)\priv}
$. 
Using sub-additivity of the sub-gaussian parameter for independent random variables, we can bound the sub-Gaussian norm of $Y_j$ as
\begin{equation}
    \sigma^2(Y_j) \leq \frac{(e^{\priv/2} - 1)}{(e^{\priv/2} + 1)\ns\priv}
\end{equation} 
We can use this to bound the maximum error of the estimator.
\begin{theorem}[Expected maximum error of RAPPOR]
\label{theo:rappor:improved}
	The expected maximum of $\hat{q}$ is given by
	\begin{equation}
		\label{eq:expected-max}
		\bEE{\norminf{\hat{q} - q}} \leq \sqrt{\frac{  2\mleft(e^{\priv /2}+1\mright)\log \ab}{\ns\left(e^{\priv /2}-1\right) \priv }} \in \bigO{\sqrt{\frac{\log\ab}{\ns\min(\priv, \priv^2)}}}
	\end{equation}
\end{theorem}
\noindent proving the main result of this subsection.

\subsection{Optimal analysis for RAPPOR}
\label{sec:tighter-analysis}
It is natural to wonder whether this $\min(\priv, \priv^2)$ dependence on the privacy parameter is order-optimal; especially as it appears in other locally private estimation tasks, such as mean estimation for high-dimensional Gaussians or product distributions~\cite{DuchiR19,AcharyaCST23}. Quite surprisingly, we will show that for frequency estimation the $\lp[\infty]$ error rate given in~\cref{eq:expected-max} is \emph{not} optimal for large values of $\priv$, and that even the simple RAPPOR algorithm can achieve significantly better.
\begin{theorem}
    \label{theo:optimal:rappor:ub}
    For $\priv \geq 1$, the expected $\lp[\infty]$ error of RAPPOR for frequency estimation satisfies
    \[
    \bEE{\norminf{\hat{q} - q}} = \bigO{\sqrt{\frac{\log\ab}{\ns e^{\priv/2}}} + \frac{\log\ab}{\ns\priv}\cdot \log\ns }\,.
    \]
\end{theorem}
Importantly, this is better than the bound in~\cref{eq:expected-max} whenever $\ns = \tilde{\Omega}((\log\ab)/\priv)$, which is the regime of interest (small constant, or vanishing, error). 
To see why this better bound may hold for large $\priv$, recall that the bound given in~\cref{eq:expected-max} relies on analyzing the expected maximum of $\ab$ (centered) Binomials random variables using their \emph{sub-gaussian} behavior. This is good when the parameters of the Binomials are not too skewed; however, in the low privacy regimes, the parameters of the Bernoulli summands become very close to $0$ (or $1$): in that case, to analyze the expected maximum of the Binomials it is tighter to see them as having a \emph{sub-gamma} behavior. Details follow.
\begin{proof}[Proof of~\cref{theo:optimal:rappor:ub}]
As alluded to above, we want to analyze the expected behavior of the maximum of Binomial random variables \emph{beyond} the sub-gaussian regime. Invoking generic bounds for sub-gamma random variables such as~\cite[Corollary~2.6]{Boucheron:13}, unfortunately, does not lend itself to the order-optimal bounds either. Instead, we rely on the ``local Glivenko--Cantelli'' bounds recently obtained by~\cite{CohenK23,BlanchardV24}, which provide a more refined upper bound: to introduce the result we will invoke, we first need some notation.
\begin{definition}
Given $\mu\in[0,1]^{\N}$, denote by $\tilde{\mu}\in[0,1/2]^{\N}$ the sequence defined by $\tilde{\mu}_i = \min(\mu, 1-\mu)$ for all $i\geq 1$, and by $\tilde{\mu}^\downarrow$ its non-increasing rearrangement. Finally, let $p_\mu$ denote the product distribution over $\{0,1\}^\N$ with mean vector $\mu$.
\end{definition}
%
%
\noindent With this in hand, the main result of~\cite{CohenK23} can be restated as follows:
\begin{theorem}[{\cite[Theorem~3]{CohenK23}}]
    \label{theo:local:gc:ck}
     Let $\ns \geq 21$, and suppose that $\mu\in[0,1]^{\N}$ is such that $\lim_{i\to\infty} \tilde{\mu}_i^\downarrow = 0$. Let $\hat{\mu}_{\ns} = \frac{1}{\ns} \sum_{j=1}^{\ns} X_j \in [0,1]^\N$ denote the empirical estimator for $\mu$ given $\ns$ \iid samples $X_1,\dots, X_{\ns}$ from $p_{\mu}$. Then the following holds:
        \[
            \bEE{\norminf{\hat{\mu}_\ns - \mu}} \lesssim
            \sup_{i\geq 1} \sqrt{\frac{\tilde{\mu}_i^\downarrow  \log(i+1)}{\ns}} +  \frac{\log \ns}{\ns} \sup_{i\geq 1}  \frac{\log(i+1)}{\log \frac{1}{\tilde{\mu}_i^\downarrow}}
        \]
\end{theorem}
\noindent(Note that~\cite{BlanchardV24} recently improved on this results by up to a $\log \ns$ factor in the second term. For simplicity, we use the slightly weaker result of~\cite{CohenK23}, as it is easier to manipulate.)\smallskip

We want to apply this result to bounding $\bEE{\norminf{\bar{Y} - \bEE{\bar{Y}}}}$. However, we face one obstacle in doing so, as we do not have a product distribution over $\{0,1\}^\ab$: while $\bar{Y}_1,\dots,\bar{Y}_{\ab}$ are indeed independent, and each of the form
\[
    \bar{Y}_i = \frac{1}{\ns}\sum_{j=1}^{\ns} Y_{i,j}
\]
where $Y_{i,1},\dots, Y_{i,\ns}$ are independent Bernoulli random variables, these Bernoullis are not identically distributed: exactly $\ns_i = \ns q_i$ of them have parameter $1-p = \frac{e^{\priv/2}}{e^{\priv/2}+1}$, and the remaining $\ns- \ns_i$ have parameter $p= \frac{1}{e^{\priv/2}+1}$. %
Of course, we do not know the $\ns_i$'s, as this is what we are trying to estimate; all we have is that
\[
    \sum_{i=1}^\ab \ns_i = \ns\,.
\]
To circumvent this issue, let us write, for each $1\leq i\leq \ab$,
\[
    \bar{Y}_i = \bar{Y}^+_i + \bar{Y}^-_i
\]
where $\ns\bar{Y}^+_i \sim \binomial{\ns_i}{1-p}$,  $\ns\bar{Y}^-_i \sim \binomial{\ns-\ns_i}{p}$ are independent. We can then express
\begin{align}
    \bEE{\norminf{\bar{Y} - \bEE{\bar{Y}}}}
    &= \bEE{\max_{1\leq i\leq \ab }\abs{\bar{Y}_i - \bEE{\bar{Y}_i}}} \notag\\
    &= \bEE{\max_{1\leq i\leq \ab }\Paren{ \abs{\bar{Y}^+_i - \bEE{\bar{Y}^+_i}}
    + \abs{\bar{Y}^-_i - \bEE{\bar{Y}^-_i}}} } \notag\\
    &\leq \bEE{\max_{1\leq i\leq \ab }\abs{\bar{Y}^+_i - \bEE{\bar{Y}^+_i}}} + \bEE{\max_{1\leq i\leq \ab }\abs{\bar{Y}^-_i - \bEE{\bar{Y}^-_i}}}
    \label{ub:two:expected:max}
\end{align}
Now, instead of taking the maximum of $\ab$ sums of Bernoullis with different parameters but same number of summands ($\ns$ summands), we take the maximum of $\ab$ sums of Bernoullis with the same parameter (\ie Binomials) but different number of summands (at most $\ns$). This does not necessarily seem like an improvement, and still does not let us apply~\cref{theo:local:gc:ck} to either of the two expectations. However, \emph{if} we could argue that ``adding summands to each Binomial'' cannot decrease the expected maximum, then we would be in good shape: that is, we want to upper bound
\[
    \bEE{\max_{1\leq i\leq \ab }\abs{\bar{Y}^+_i - \bEE{\bar{Y}^+_i}}}
\]
by
\[
\bEE{\max_{1\leq i\leq \ab }\abs{Z^+_i - \bEE{Z^+_i}}}
\]
where $\ns Z^+_i \sim \binomial{\ns}{1-p}$ (instead of $\binomial{\ns_i}{1-p}$). Intuitively, this seems reasonable, as adding independent summands should make the Binomial more likely to deviate from its expectation. The next lemma makes this intuition rigorous:
\begin{lemma}
    \label{lemma:coupling:binomials}
    Fix $n_1,\dots, n_k, m_1,\dots, m_k \in \N$ and $p_1,\dots, p_k\in[0,1]$, with $n_i \leq m_i$ for all $i$. Let $N_1,\dots, N_k$ and $M_1,\dots, M_k$  be (not necessarily independent) random variables with $N_i \sim\binomial{n_i}{p_i}$ and $M_i \sim\binomial{m_i}{p_i}$. Then
    \[
        \expect{\max_{1\leq i\leq k} \abs{N_i - \expect{N_i}}} \leq \expect{\max_{1\leq i\leq k} \abs{M_i - \expect{M_i}}}\,.
    \]
\end{lemma}
\begin{proof}
    Set $\tilde{N}_i \eqdef N_i - \expect{N_i}$ and $\tilde{M}_i \eqdef M_i - \expect{M_i}$ for all $1\leq i\leq k$. 
    We give a coupling of $\tilde{N}_i,\tilde{M}_i$ such that
    \[
        \expectCond{\tilde{M}_i}{\tilde{N}_i} = \tilde{N}_i\,.
    \]
    Such a coupling can be obtained by setting $\tilde{M}_i = \tilde{N}_i + \Delta_i - \expect{\Delta_i}$, where $\Delta_i\sim\binomial{m_i-n_i}{p_i}$ is independent of $\tilde{N}_i$. Then it is easy to check that $\tilde{M}_i$ has the right distribution, since  $N_i+\Delta_i\sim\binomial{m_i}{p_i}$; and the the conditional expectation is indeed as claimed. 
    Using this coupling, we obtain
    \begin{align*}
        \expect{\max_{1\leq i\leq k} |\tilde{N}_i|}
        &= \expect{\max_{1\leq i\leq k} |\expectCond{\tilde{M}_i}{\tilde{N}_i}|} \\
        &\leq \expect{\expectCond{\max_{1\leq i\leq k} |\tilde{M}_i|}{\tilde{N}_i}} \tag{Jensen's inequality}\\
        &= \expect{\max_{1\leq i\leq k} |\tilde{M}_i|}
    \end{align*}
    establishing the lemma.\footnote{More generally, via the existence of this coupling the argument shows that $(\tilde{N}_1,\dots, \tilde{N}_k) \preceq_{\rm{}cx} (\tilde{M}_1,\dots, \tilde{M}_k)$ (domination in the convex order), which in turn is equivalent to having $\expect{\phi(\tilde{N}_1,\dots, \tilde{N}_k)}\leq \expect{\phi(\tilde{M}_1,\dots, \tilde{M}_k)}$ for every convex function $\phi$.}
\end{proof}
Let $Z^+_1,\dots,Z^+_k$ (resp. $Z^-_1,\dots,Z^-_k$) be \iid\ with $\ns Z^+_i\sim \binomial{\ns}{1-p}$ random variables (resp. $\ns Z^-_i\sim\binomial{\ns}{p}$). Invoking~\cref{lemma:coupling:binomials} with $m_1=\dots=m_k= \ns$ separately on the two expectations of~\cref{ub:two:expected:max}, we get
\[
    \bEE{\norminf{\bar{Y} - \bEE{\bar{Y}}}}
    \leq \bEE{\max_{1\leq i\leq \ab }\abs{Z^+_i - \bEE{Z^+_i}}} + \bEE{\max_{1\leq i\leq \ab }\abs{Z^-_i - \bEE{Z^-_i}}}
\]
Both of the terms in the RHS now fit the setting of~\cref{theo:local:gc:ck}. Moreover, since their parameters are $p$ (for the first expectation) and $1-p$ (for the second), in both case the corresponding $\tilde{mu}_i=\min(p,1-p) = p = \frac{1}{e^{\priv/2}+1}$ is the same, and applying the theorem will give the same upper bound for both expectations. Thus,~\cref{theo:local:gc:ck} yields
\[
    \bEE{\norminf{\bar{Y} - \bEE{\bar{Y}}}}
    \lesssim \sqrt{\frac{p \log(\ab+1)}{\ns}} + \frac{\log\ns}{\ns} \cdot \frac{\log(\ab+1)}{\log\frac{1}{p}}
\]
Recalling the setting of $p$, along with the fact that $\bEE{\norminf{\hat{q} - q}} = \frac{e^{\priv/2}+1}{e^{\priv/2}-1}\bEE{\norminf{\bar{Y} - \bEE{\bar{Y}}}}$, finally gives
\begin{align}
    \bEE{\norminf{\bar{Y} - \bEE{\bar{Y}}}}
    &\lesssim \sqrt{\frac{\log(\ab+1)}{\ns}\cdot \frac{e^{\priv/2}+1}{(e^{\priv/2}-1)^2}} + \frac{\log\ns}{\ns} \cdot \frac{e^{\priv/2}+1}{e^{\priv/2}-1} \cdot \frac{\log(\ab+1)}{\log(e^{\priv/2}+1)}
\end{align}
To conclude, observe that the RHS is $\bigO{\sqrt{\frac{\log\ab}{\ns \priv^2}} + \frac{\log\ab}{\ns\priv}\cdot \log\ns}$ for $\priv \leq 1$, and 
$
\bigO{\sqrt{\frac{\log\ab}{\ns e^{\priv/2}}} + \frac{\log\ab}{\ns\priv}\cdot \log\ns }
$
for $\priv \geq 1$.
\end{proof}

\subsection{Upper bound for Projective Geometry Response}
\label{sec:pgr}
\input{sec-pgr-analysis.tex}

\section{Worst-case, information-theoretic lower bounds}
\label{subsec:lower-bounds}
\input{sec-lowerbound}

\section{Amplification by shuffling}
\label{sec:shuffle}
\input{amplification-by-shuffling}

\section{Empirical evaluation and findings}
\label{sec:empirical-findings}
\begin{figure}[h]
{\centering
    \includegraphics[width=0.99\textwidth]{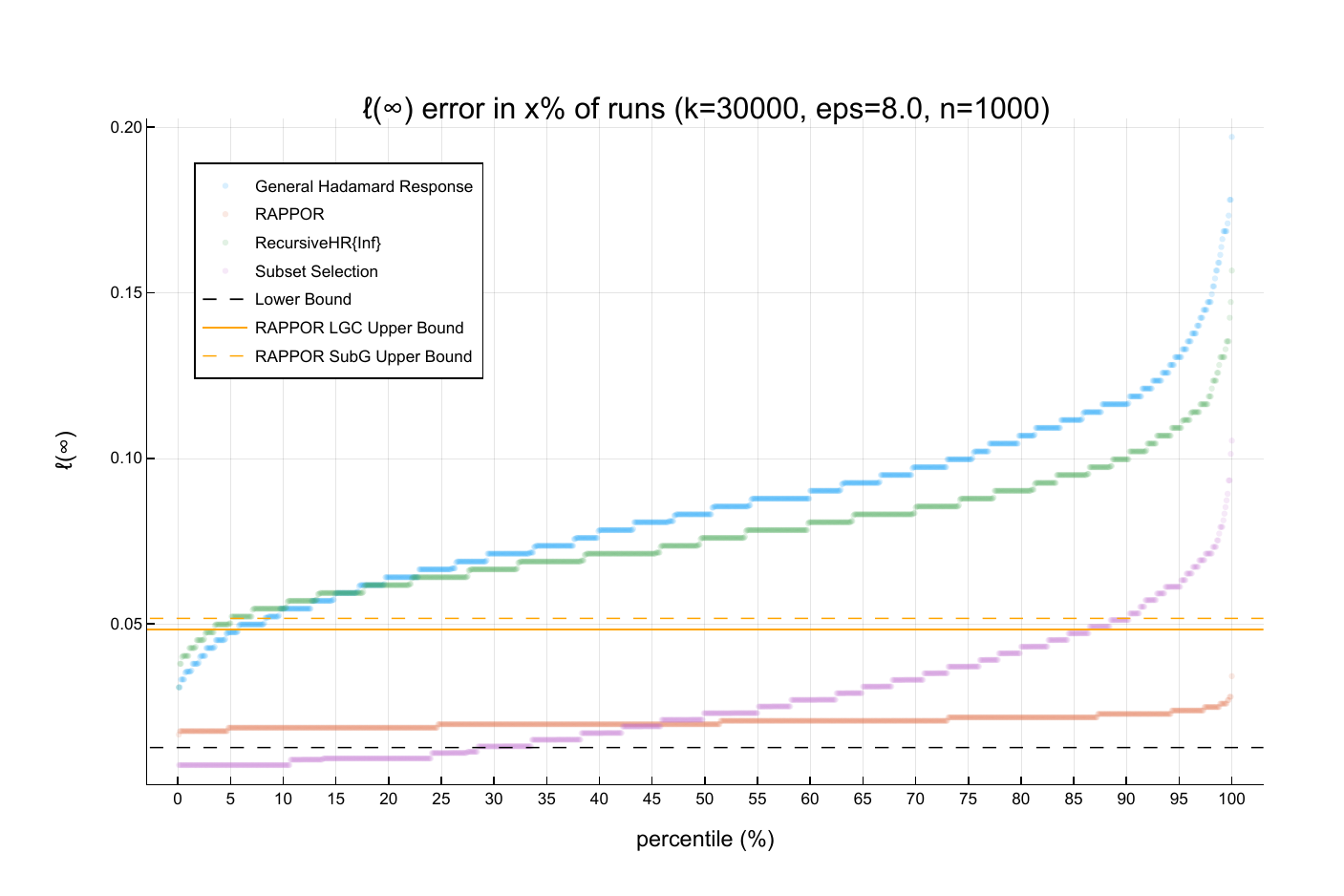}
    \caption{Maximum error ($\lp[\infty]$) in $x\%$ of 1000 runs. $\priv=5$, $\ab=5000$, $\ns=2000$, and the distribution is a point-mass at the first index. Horizontal lines indicate upper and lower bounds on the expected maximum explored in this paper. The improved bound for RAPPOR is the one given by use of the exact sub-Gaussian parameter.}
    \label{fig:percentiles}
}
\end{figure}
We performed empirical evaluations of various protocols in order to evaluate which are likely to have in practice improved upper bounds on their theoretical worst-case error, and to distinguish between constant factors in their performance. Additionally we were interested in whether there were \emph{distribution-dependent} factors which influence $\lp[\infty]$ error.

\Cref{fig:percentiles} is a plot of $\lp[\infty]$ in $x\%$ of runs, effectively representing a CDF of the error distribution. Horizontal lines indicate bounds discussed in this paper. Optimal protocols should have their mean close to the lower bound. This implies much stronger upper bounds can likely be derived for Subset Selection.

\input{concentration-figures.tex}

All of the experiments together took approximately 4 hours on a 2022 M2 Macbook air.

\section{Discussion and future work}
\label{sec:discussion}

\paragraph{The logarithmic factor in the upper bound.}
We prove our upper bounds using the result of~\cite{CohenK23}, which is where the $\log\ns$ term appears. The authors conjectured that this term could be removed in general, but follow--up work~\cite{BlanchardV24} demonstrated by counter--example that this is not the case. They show that there is a distribution--dependent interpolation between the term being $T(\ns)/\ns$ and $T(\ns)\log\ns/\ns$. This logarithmic factor has implications for understanding which error regime is dominating, given the parameters of the algorithm $(\priv,\ab,\ns)$. While it is possible to imagine that a careful application of the tools in~\cite{BlanchardV24} could resolve this matter in general or for a specific use--case, in the meantime we emphasise that empirical analysis is still crucial.

\paragraph{Tighter analysis of other protocols.} 
While we analyse RAPPOR which admits a simple analysis due to the independence of coordinates, and Projective Geometry Response which represents the state of the art in low-communication LDP protocols, we believe the tools introduced in this paper are applicable in a very general way to most LDP protocols. Of particular interest would be Subset Selection, which has optimal mutual information between inputs and outputs of the local randomiser, and the two Hadamard Response protocols that we include in our empirical analysis. 

\paragraph{Histograms in the shuffle model.} While our shuffle DP result (\cref{theo:pgr:shuffle}) yields better error, communication, and number of messages, it does have a limitation on the parameter range, namely $\ns = \Omega(\log(1/\delta)/\priv^2)$. It would be interesting to weaken this requirement to match the central DP one (where the dependence on $\priv$ is only linear). Moreover, one could hope to further improve the resulting error to the central DP bound (\ie a $\min(\log\ab,\log(1/\delta))$ dependence instead of $\sqrt{\log\ab\log(1/\delta)}$), or, alternatively, prove a matching lower bound separating the two models. %

\paragraph{Broader applicability of lower bound techniques.} While we prove lower bounds for the specific case of $\lp[\infty]$ frequency estimation, the tools used are extremely general and we imagine that their application could provide new lower bounds for a variety of problems in the LDP setting, especially with the low--privacy regime in mind.

\ifidentified %
\paragraph{Acknowledgments.} The authors would like to thank Guy Blanc for the proof of~\cref{lemma:coupling:binomials}, and Albert Cheu for insightful discussions regarding the use of amplification by shuffling (\cref{lemma:amp:shuffling}).
\fi
\clearpage
\bibliographystyle{alpha}
\bibliography{references-neurips.bib}

\appendix

\section{Deferred miscellaneous proofs}
\label{appendix:misc}

\begin{proof}[Proof of~\cref{lemma:rappor-unbiased}]
\label{proof:rappor-unbiased}
Note that each $\bar{Y}_j$ (for $1\leq j\leq \ab$) is the normalised sum of $\ns$ independent Bernoulli random variables $Y_{1,j},\ldots, Y_{\ns,j}$ with parameter either $\frac{e^{\priv/2}}{e^{\priv/2}+1}$ (for a fraction $q_j$ of them) or parameter $\frac{1}{e^{\priv/2}+1}$ (for the $1-q_j$ remainders). It follows that, for every fixed $j\in[\ab]$,
\[
    \bEE{\bar{Y}_j} = q_j\frac{e^{\priv/2}}{e^{\priv/2}+1} + (1-q_j) \frac{1}{e^{\priv/2}+1} = \frac{e^{\priv/2}-1}{e^{\priv/2}+1}q_j + \frac{1}{e^{\priv/2}+1}
\]
from which we derive our unbiased estimator,
\[
    \bEE{\hat{q}_j} = \frac{e^{\priv/2}+1}{e^{\priv/2}-1}\bEE{Y_j}- \frac{1}{e^{\priv/2}-1}.\qedhere
\]
\end{proof}

\begin{proof}[Proof of~\cref{prop:generic:transformation}]
    Let $A$ be as in the assumptions of the proposition. For any $\priv >0$, let $T \eqdef  \min(\clg{\priv}, L)$ and $\priv' \eqdef \priv/T \in(0,1]$. Each user uses their randomiser $T$ times with privacy parameter $\priv'$, independently, on their (single) input, and sends the $T$ results to the server as if they were $T$ distinct users. This simulates the protocol $A$ on $\ns' \eqdef T\ns$ users, but does not affect the frequency of each element (as each count is multiplied by $T$ but then normalised by $T\ns$ instead of $\ns$). The claims about number of messages and per-user communication are immediate; as for the error, note that we have
    \[
    \mathcal{E}\left(\ns',\ab,\priv'\right)
     = \bigO{\sqrt{\frac{\log\ab}{\ns'\priv'^2}}}
     = \bigO{\sqrt{\frac{\log\ab}{\ns T\min(1,\priv^2/T^2)}}}
     = \bigO{\sqrt{\frac{\log\ab}{\ns\min(T,\priv^2/T)}}}
    \]
    and the result follows from a distinction of cases, observing that 
    (a)~if $\priv < 1$, then $T=1$, and $\min(T,\priv^2/T) = \priv^2 = \min(L, \priv, \priv^2)$; (b)~if $1\leq \priv \leq L$, then $T=\clg{\priv}$, and 
    $
    \min(T,\priv^2/T) = \min(\clg{\priv}, \priv^2/\clg{\priv}) = \Theta(\priv) = \Theta(\min(L,\priv,\priv^2)
    $; and finally (c)~if$\priv > L$, then $T=L$ and $\min(T, \priv^2/T) = L = \min(L,\priv,\priv^2)$.
    (Where for case (b) we used that $\frac{\priv}{\clg{\priv}} = \priv$ if $\priv < 1$, and otherwise is in $[1/2,1]$.)
\end{proof}

Finally, we prove the $(\ddagger)$ bound from~\cref{tab:previous-results}: that is, that the unnamed LDP protocol from~\cite[Theorem 3.1]{ChenKO23} achieves worst-case expected $\lp[2]^2$ error at most $\frac{4\ab\ns e^{2\priv/\numbits}}{(e^{\priv/\numbits}-1)^2}$:
\begin{proof}[Proof of this claim]
    Starting from~\cite[Appendix E.4]{ChenKO23}, where the authors bound the expected $\lp[\infty]$ error, and specifically their Eq.~(23), we have that the output $\hat{q}$ of their LDP protocol (and the true histogram $q$) satisfy, for every $j\in[\ab]$,
    \[
     \hat{q}_j - q_j = \frac{1}{\ns\numbits} \sum_{i=1}^\ns \sum_{m=1}^{\numbits} \Paren{\hat{X}^{(m)}_i(j) - X_i(j)}
    \]
    using their notation for $\hat{X}^{(m)}_i$ (the $m$-bit of the privatized message sent by user $i$). Using their almost sure bound that states that
    \[
        \abs{\hat{X}^{(m)}_i(j) - X_i(j)} \leq \frac{e^{\priv/\numbits}+1}{e^{\priv/\numbits}-1} + 1 = \frac{2e^{\priv/\numbits}}{e^{\priv/\numbits}-1}
    \]
    for all $m\in[\numbits], j\in[\ab], i\in[\ns]$, we get
    \begin{align*}
     \bEE{\normtwo{\hat{q}-q}^2} &=
     \sum_{j=1}^\ab\bEE{\Paren{\hat{q}_j-q_j}^2} \\
     &\leq \frac{1}{(\ns\numbits)^2} \sum_{j=1}^\ab \sum_{i=1}^\ns \sum_{m=1}^{\numbits}\sum_{i'=1}^\ns \sum_{m'=1}^{\numbits} \bEE{\Paren{\hat{X}^{(m)}_i(j) - X_i(j)}\Paren{\hat{X}^{(m')}_{i'}(j) - X_{i'}(j)}} \\
     &= \frac{1}{(\ns\numbits)^2} \sum_{j=1}^\ab \sum_{i=1}^\ns \sum_{m=1}^{\numbits} \bEE{\Paren{\hat{X}^{(m)}_i(j) - X_i(j)}^2} \tag{$\ast$}\\
     &\leq \frac{\ab}{\ns\numbits} \frac{4e^{2\priv/\numbits}}{\Paren{e^{\priv/\numbits}-1}^2}
    \end{align*}
    where $(\ast)$ relies on the facts that (1)~$\hat{X}^{(m)}_i(j)$ is an unbiased estimator of $X_i(j)$, and (2)~that the $\hat{X}^{(m)}_i(j)$'s are independent across both $m\in[\numbits]$ and $i\in[\ns]$.
\end{proof}

%

\section{More empirical analysis and bounds plots}
\begin{figure}[h]
{\centering
    \includegraphics[width=0.99\textwidth]{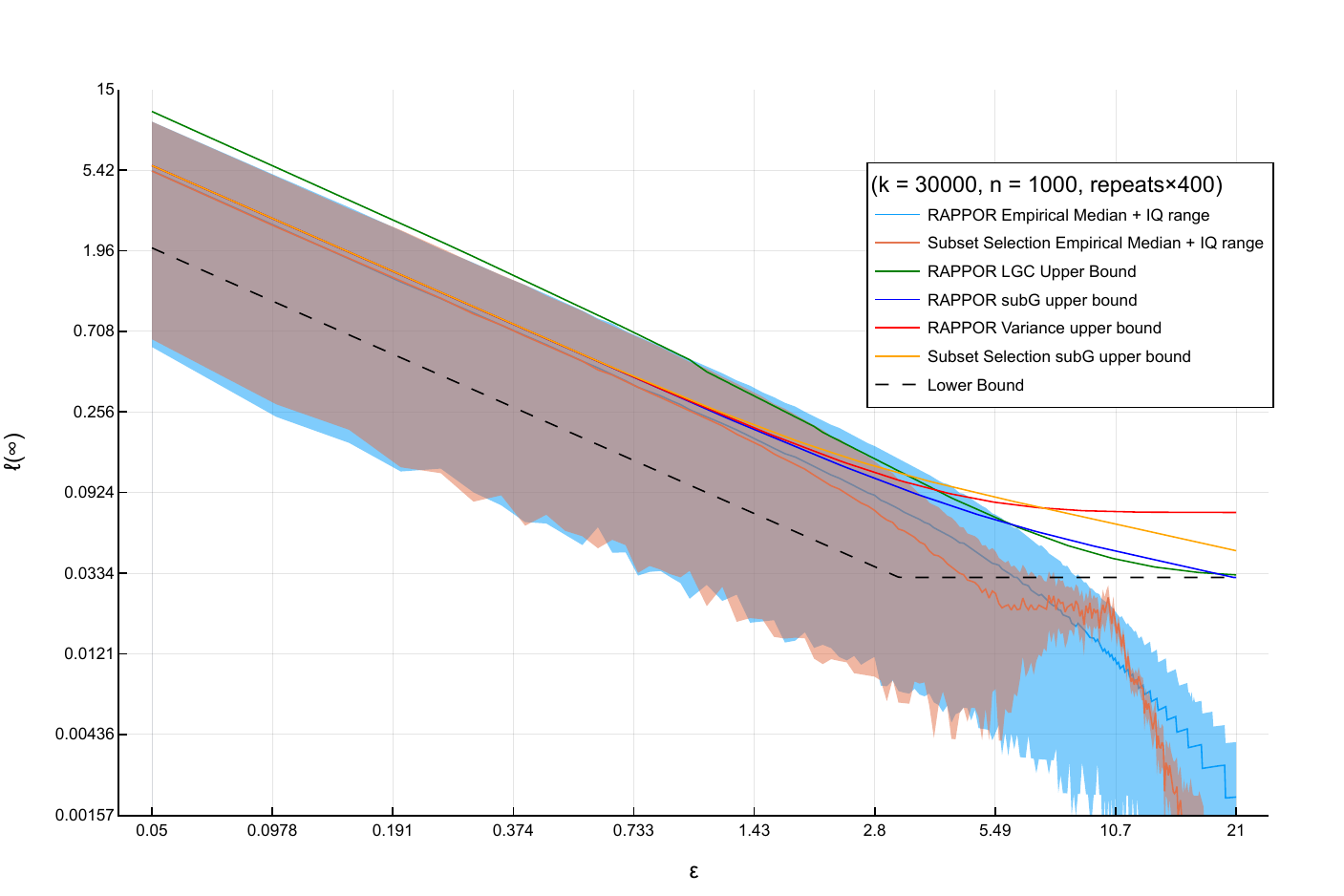}
    \caption{Log--log plot of median $\lp[\infty]$ error with upper and lower quartiles, by $\priv$. Lower and upper bounds discussed in this work included for comparison. No normalisation or clipping has been applied leading to $\lp[\infty]>2$ in the high privacy regime.}
    \label{fig:error-by-eps}
}
\end{figure}

\Cref{fig:error-by-eps} demonstrates RAPPOR and Subset Selections' performance against its theoretical upper and lower bounds. One can see subset selection become worse in the regime where $\priv>\log\ab$, until the subset size reaches 1 and it becomes practically non-private frequency estimation. We can see a minor transition in the low-privacy regime to the Local Glivenko--Cantelli behaviour, which is much more apparent for large $\ab$ as  while simple analysis of the sub--Gaussian parameters begins to diverge from the lower bound. \Cref{fig:only-bounds} plots only the theoretical bounds to demonstrate this more clearly.
\begin{figure}[h]
{\centering
    \includegraphics[width=0.95\textwidth]{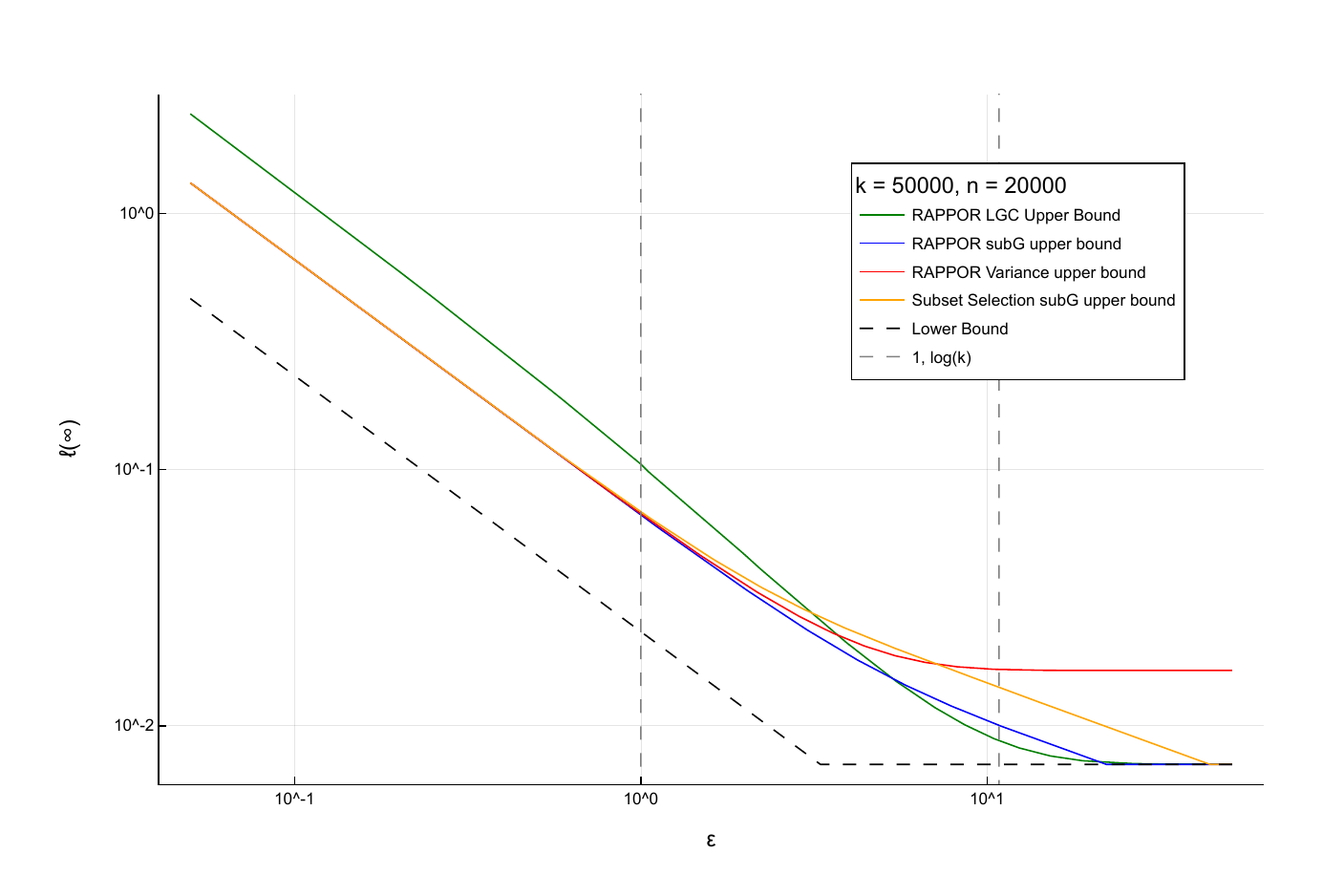}
    \caption{Log--log plot of some bounds discussed in this paper, demonstrating a transition to the local Glivenko Cantelli bound in the intermediate privacy regime. This privacy regime was too computationally expensive to simulate, however larger values of $\ab$ seem to be of interest as the shuffle model allows for larger values of $\priv$ to compensate for the logarithmic error loss in the alphabet size. In addition some literature suggests that domain reduction tools may introduce an impractical amount of noise~\cite{erlingsson2020}.}
    \label{fig:only-bounds}
}
\end{figure}
%
%

\end{document}

%% file: sec-preliminaries.tex
\paragraph{(Local) differential privacy.} We first recall the definition of local differential privacy (LDP):
\begin{definition}[Locally private randomiser]
    An algorithm $Q\colon\mathcal{X}\rightarrow\mathcal{Y}$ satisfies $\priv$-differential privacy if for all pairs of inputs $x,x'\in\mathcal{X}$, all sets of outputs $S\subseteq\mathcal{Y}$ are $\priv$-close,
    \begin{equation*}
        \probaOf{Q(x)\in S } \leq e^\priv\probaOf{Q(x')\in S }.
    \end{equation*}
\end{definition}
A \emph{locally private protocol} is typically a pair of algorithms $Q\colon\mathcal{X}\rightarrow\mathcal{Y}$, executed by the user on their data, and $A\colon\mathcal{Y}^n\rightarrow\mathcal{X}^\ab$, executed by the server to transform the randomised outputs into an unbiased estimator for the true quantity.
It is well known (see, \eg~\cite{DJW:13}) that all LDP Frequency Estimation algorithms have a stochastic matrix representation which maps every element in the input alphabet to a probability distribution over the output alphabet. Furthermore,~\cite{JMLR:v17:15-135} demonstrated that all optimal mechanisms obey a ``binary'' property, in that their stochastic matrix only contains probabilities weighted by either $e^\priv$ or $1$ (appropriately normalised). Recently another ``family of LDP protocols'' was introduced by~\cite{ASZ:18:HR} and used in~\cite{feldman22a}, where each input is associated with a set of high probability outputs determined by some useful set system.

\paragraph{Frequency Estimation.}
\label{sec:frequency-estimation}
Under frequency estimation, the aim is to estimate the empirical frequency of the observations. We denote the dataset of $\ns$ observations as $X^\ns \eqdef\{x_1,\dots,x_\ns\}$ and their empirical frequencies as the probability vector $q\eqdef q(X^\ns)$, where $q_i\eqdef\frac{1}{\ns}\sum_{j=1}^\ns \indic{x_j=i}$. That is,
$q$ is an element of the probability simplex, 
\[
	\Delta_\ab\eqdef\setOfSuchThat{ p\in\R^\ab_{\geq 0} }{ \sum\limits_{i=1}^k p_i=1 }
\] 
A \emph{frequency estimator} over an alphabet $[\ab]$ is a (possibly randomized) function $\mathcal{A}\colon[\ab]^\ns\to\Delta_\ab$, which approximates the true frequencies sufficiently well. The quality of the resulting approximation, denoted $\hat{q}$, is measured through a suitably chosen loss function.
Typical choices of distance are the $\lp[r]$ norms, for $r\in[1,\infty]$; in this paper, we will be mostly concerned with the $\lp[\infty]$ distance, where $d(q,\hat{q}) = \norminf{q-\hat{q}} = \max_{1\leq i \leq \ab} \abs{q_i-\hat{q}_i}$, and consider the \emph{expected loss} $\bEE{\norminf{q-\hat{q}}}$ (where the expectation is over the randomness of the frequency estimator).

\paragraph{Distribution learning (estimation).}
Given independent and identically distributed (\iid) samples from an unknown probability distribution $p$, the goal of distribution learning (density estimation) is to find a distribution $\hat{p}$ that approximates $p$ sufficiently well. As before, the quality of the approximation is typically measured through a suitably chosen loss function, quantifying the distance between the true distribution $p$ and its approximation $\hat{p}$. From the above, one can see that the key difference between frequency and distribution (under the same loss function) is that in the former, observations are arbitrary, while in the latter they are assumed \iid from some common underlying probability distribution. As such, the expected loss also takes into account the randomness of the samples themselves.

\paragraph{Frequency Estimation implies Distribution Estimation.}
\label{sec:distributionlearning:frequency}
%
%
For any given distance measure $d$ we denote the expected sampling error (under loss $d$) between the empirical histogram $\hat{p}_{\ns}$ from $\ns$ \iid samples and the true underlying distribution $p$ as $\Phi(d, \ab,\ns)$:
\[
	\Phi(d, \ab,\ns) = \sup_{p\in\Delta_\ab} \bEE{d(p, \hat{p}_{\ns})}
\]
Further,  
denote the optimal worst-case (over $p$) expected error under loss $d$ from $\ns$ \iid samples under $\priv$-LDP constraints as $\Phi^\text{priv}\eqdef\Phi^\text{priv}(d,\priv,\ab,\ns)$. The ``optimal'' here refers to quantifying over all locally private distribution estimators. We can analogously express the optimal expected loss when learning the empirical frequencies as $\varphi^\text{priv}\eqdef\varphi^\text{priv}(d,\priv,\ab,\ns)$.
\begin{fact}
    \label{claim:learning:v:frequency}
	Accurate \emph{frequency} estimation implies accurate \emph{distribution} estimation: namely, 
	$
	\Phi^\mathrm{priv}(d,\priv,\ab,\ns)\leq\Phi(d,\ab,\ns)+\varphi^\mathrm{priv}(d,\priv,\ab,\ns)
	$.
\end{fact}
For specific norms this allows us to derive (1)~a lower bound on $\varphi^\text{priv}$ (locally private frequency estimation) from a lower bound on $\Phi^\text{priv}$ (locally private distribution estimation), and (2)~an upper bound on $\Phi^\text{priv}$ from an upper bound on $\varphi^\text{priv}$, recalling that (see, for instance,~\cite{kamath2015,canonne2020short})
\begin{align*}
\Phi(d,\ab,\ns)=
	\begin{cases}
		\bigTheta{\sqrt{\frac{\ab}{\ns}}} &\text{ when } d=\lp[1]\\
		\bigTheta{\frac{1}{\sqrt{\ns}}}&\text{ when }d=\lp[2],\lp[\infty]\\
	\end{cases}
\end{align*}

%% file: sec-pgr-analysis.tex
Projective geometry response (PGR), introduced by~\cite{feldman22a} achieves optimal rates for $\lp[2]^2$ error, communication, and near-optimal processing time. The protocol is based on the general template established in~\cite{ASZ:18:HR}: specifically, PGR relies on a set structure defined by \emph{projective planes}, as detailed next. For a prime power $\pp$ and $\ab = \frac{\pp^t - 1}{\pp-1}$, the authors define a $t$-dimensional vector space $\mathbb{F}^t_\pp$, where each element  $x\in[\ab]$ is represented by one of the canonical basis vectors. These basis vectors in turn each uniquely determine a projective plane $S(x)$, such that there are $s=|S(x)|=\frac{\pp^{t-1} - 1}{\pp-1}$ ``high probability'' elements. Every one of these sets in turn has an intersection with every other set of size $c=|S(x)\cap S(x')|=\frac{\pp^{t-2} - 1}{\pp-1}$. By choosing a prime power $\pp\approx e^\priv + 1$, optimal error is achieved.

\begin{theorem}
    \label{theo:optimal:pgr:ub}
    Projective geometry response (PGR)~\cite{feldman22a} achieves the optimal rate for $\lp[\infty]$ error. More specifically, the expected $\lp[\infty]$ error of Projective Geometry Response for frequency estimation satisfies
    \begin{align*}
        \bEE{\norm{\hat{q}-q}_\infty}&\leq\sqrt{\frac{16(2e^\priv + 1)^2\log(\ab+1)}{e^\priv(e^\priv - 1)^2\ns}}+\frac{4(2e^\priv + 1)\log(\ab+1)}{(e^\priv - 1)\priv\ns}\log\ns
\\
            &\in\bigO{\sqrt{\frac{\log\ab}{ne^\priv}}+\frac{\log\ab}{n\priv}\log\ns}
    \end{align*}
\end{theorem}
We briefly recall some details of PGR as described in~\cite{feldman22a}. We will make use of the size of each subset $s$ and the size of each intersection $c$ as described in the introduction to this section, as well as the probability of returning any element $y$ from the output alphabet, given $x$ as an input:
\begin{equation}
    \label{eq:pgr:probabilities}
    Q(Y=y \mid X=x) = \begin{cases}
        \frac{e^\priv}{se^\priv + \ab - s}&\text{ if }y\in S(x)\\
        \frac{s-c}{se^\priv + \ab - s}&\text{ otherwise.}
    \end{cases}
\end{equation}
The estimate $\hat{q}$ of the frequency vector is then given by
\begin{equation}
    \label{eq:pgr:estimate}
    \hat{q}_x \eqdef \alpha\cdot \frac{1}{\ns} \sum_{i=1}^{\ns} \indicSet{Y_i \in S(x)} + \beta, \qquad x\in[\ab]
\end{equation}
where
\begin{equation}
    \label{eq:pgr:alphabeta}
    \alpha = \frac{(e^\priv-1)s+\ab}{(e^\priv-1)(s-c)}\,,\qquad \beta = -\frac{(e^\priv-1)c+s}{(e^\priv-1)(s-c)}
\end{equation}
so that $\hat{q}$ is an unbiased estimator of the true frequency vector $q$. 
 Since a prime power can be found within a factor $2$ of any number, we can choose $\pp$ such that
\begin{equation}
    \label{eq:nearby-prime}
    e^\priv + 1\leq \pp\leq 2(e^\priv + 1)
\end{equation} 
(While we could instead choose $\pp<e^\priv + 1$ and set the inequality to be an undershoot by a factor two, an investigation in that direction led to worse constants and trickier analysis.) 

\noindent Recall that, for every integer $\ell$, 
\begin{equation}
    \pp^\ell - 1=(\pp-1)(1+\pp+\ldots+\pp^{\ell-1})\,,\label{eq:int-divisibility}
\end{equation}
an identity we will rely on extensively in the rest of the section. 

\begin{remark}
    Note that, as introduced above, we must have $\ab = \frac{\pp^t - 1}{\pp-1}$ for a prime power $\pp$ satisfying~\eqref{eq:nearby-prime}. Other values of $\ab$ must be rounded up, losing up to a factor $O(e^\priv)$ in the domain size (\ie working instead with a domain size $\ab' = O(e^\priv k)$). We hereafter ignore this detail, which does not affect the final bound of~\cref{theo:optimal:pgr:ub} unless $\priv \gg \log \ab$; and assume that $\ab$ is of the form stated above. For the same reason, we additionally can assume $t\geq 3$. 
\end{remark}
\begin{proof}[Proof of~\cref{theo:optimal:pgr:ub}]
We start with two lemmas which will be useful in proving the optimal error rate.
\begin{lemma}
\label{lem:s-larger}
    The common size $s$ of every subset $S(x)$, $x\in[\ab]$, satisfies $s\geq e^\priv+2$.
\end{lemma}
\begin{proof}
We can rewrite $s$ as
\begin{equation*}
    s=\frac{\pp^{t-1}-1}{\pp-1}=\sum\limits_{i=0}^{t-2}\pp^i
\end{equation*}
As $t\geq 3$, we have, by applying~\cref{eq:nearby-prime,eq:int-divisibility}
\begin{equation*}
    s=\frac{\pp^{t-1} - 1}{\pp-1}=\sum\limits_{i=0}^{t-2}\pp^i\geq 1+\pp\geq e^\priv + 2.\qedhere
\end{equation*}
\end{proof}
Now remembering that whatever expected error we compute will be multiplied by the normalising constant $\alpha$, we would like to bound that in terms of $\priv$.
\begin{lemma}
\label{lem:alpha-bound}
The value $\alpha$ defined in~\eqref{eq:pgr:alphabeta} satisfies 
    $\alpha\leq 2+\frac{2+\pp}{e^\priv - 1}\leq2+\frac{4+2e^\priv}{e^\priv - 1}$.
\end{lemma}
\begin{proof}
    First note that 
    \begin{equation*}
        \alpha = \frac{(e^\priv - 1)s + k}{(e^\priv-1)(s-c)}=\frac{s}{(s-c)}+\frac{k}{(e^\priv - 1)(s-c)}
    \end{equation*}
    and note that by simple application of~\cref{eq:int-divisibility} we get
    $
        s-c=\pp^{t-2}
    $. 
    Applying the identity again we get:
    \begin{equation}
        \frac{s}{\pp^{t-2}}=\sum_{i=0}^{t-2}\frac{1}{\pp^{i-t+2}}=\frac{1}{\pp^{t-2}}+\frac{1}{\pp^{t-1}}+\ldots +1\leq 2\label{eq:s-sum}
    \end{equation}
    where the final inequality comes from the fact that with $\pp\geq 2$ this is bounded by the geometric series.
    For $k$ we have the same series, with the addition of a single $\pp$ term.
    \[
        \frac{k}{\pp^{t-2}}=\sum_{i=0}^{t-1}\frac{1}{\pp^{i-t+2}}=\frac{1}{\pp^{t-2}}+\frac{1}{\pp^{t-1}}+\ldots +1+\pp
    \]
    The result is then at most $2+\pp\leq 2+2(e^\priv + 1)$, by applying~\cref{eq:nearby-prime}.
\end{proof}
We are now ready to apply~\cref{theo:local:gc:ck}. As we did for RAPPOR, we will break the expected maximum into the sum of error over two vectors of binomials $Z^+=\binomial{\ns}{p^+}$ and $Z^-=\binomial{\ns}{p^-}$ where, recalling~\eqref{eq:pgr:probabilities},
\begin{equation}
    \label{eq:pgr:probabilities:p+p-}
    p^+=\frac{e^\priv}{se^\priv + \ab -s}\,,\qquad p^-=\frac{s-c}{se^\priv + \ab - s}
\end{equation}
In this case we do not have that $p^+=1-p^-$ so we will need to bound the following
\begin{equation*}
    \bEE{\norm{\hat{q}-q}_\infty}\leq\alpha\left(\sqrt{\frac{\log(k+1)}{n}}\Paren{\sqrt{p^+}+\sqrt{p^-}}+\frac{\log(k+1)}{n}\log(n)\left(\frac{1}{\log(1/p^+)}+\frac{1}{\log(1/p^-)}\right)\right).
\end{equation*}
We briefly emphasize that while RAPPOR has clear independence between coordinates, the result of~\cref{theo:local:gc:ck} due to~\cite{CohenK23} does not require independence; therefore we can apply it to PGR and other subset-based protocols to determine their expected $\lp[\infty]$ error.
 
First we will upper bound $\sqrt{p^+}+\sqrt{p^-}$:
\begin{lemma}
For $p^+,p^-$ defined in~\eqref{eq:pgr:probabilities:p+p-}, we have
    $p^+ + p^-\leq 2/e^\priv$,
    and so $\sqrt{p^+}+\sqrt{p^-} \leq \sqrt{4/e^{\priv}}$.
\end{lemma}
\begin{proof}
    First, $p^+$ is clearly less than $\frac{1}{s}$ when we remove $k-s$ from the denominator, and so is at most $1/(e^\priv + 2)\leq 1/e^\priv$ by~\cref{lem:s-larger}. For the other term, notice that by~\cref{eq:s-sum} we have $s/(s-c) \geq 1$, from which,
    \begin{align*}
        p^-\leq\frac{s-c}{se^\priv}\leq\frac{1}{e^\priv}
    \end{align*}
    Overall, we get that $p^+ + p^-\leq 2/e^\priv$. The conclusion follows from the AM-GM inequality, as
    $\sqrt{p^+}+\sqrt{p^-} \leq \sqrt{2(p^+ + p^-)}$. 
\end{proof}
\noindent Next we bound $\log^{-1}(1/p^+)+\log^{-1}(1/p^-)$:
\begin{lemma}
    We have $\frac{1}{\log(1/p^+)}+\frac{1}{\log(1/p^-)}\leq \frac{2}{\priv}$.
\end{lemma}
\begin{proof}
We will proceed in both cases by lower-bounding the denominators. First,
    \begin{align*}
        \log(1/p^-) = \log\left(\frac{se^\priv + k - s}{s-c}\right)\geq\log\left(\frac{se^\priv}{s-c}\right)\geq\log(e^\priv)=\priv.
    \end{align*}
Next,
\begin{align}
    \log(1/p^+) = \log\left(\frac{se^\priv + k -s}{e^\priv}\right)&\geq\log\left(s\right)&\nonumber\\
    &\geq\log(e^\priv + 2) \tag{\cref{lem:s-larger}}\\
    &\geq \priv&\nonumber
\end{align}
As such adding the reciprocals of both terms give a bound of $2/\priv$.
\end{proof}
The only step left to bound the $\lp[\infty]$ error of PGR is to multiply through these bounds and that on $\alpha$ (\cref{lem:alpha-bound}) to get the final error bound. Doing so gives
\begin{align*}
    \bEE{\norm{\hat{q}-q}_\infty}&\leq\sqrt{\frac{16(2e^\priv + 1)^2\log(\ab+1)}{e^\priv(e^\priv - 1)^2\ns}}+\frac{4(2e^\priv + 1)\log(\ab+1)}{(e^\priv - 1)\priv\ns}\log\ns
\end{align*}
For the low-privacy regime this gives
\begin{equation*}
    \bEE{\norm{\hat{q}-q}_\infty}\in\bigO{\sqrt{\frac{\log\ab}{ne^\priv}}+\frac{\log\ab}{n\priv}\log\ns}\,
\end{equation*}
as claimed.
\end{proof}

%% file: sec-lowerbound.tex
We will follow the ``chi-squared lower bound'' framework of~\cite{AcharyaCT:IT1} to obtain our information-theoretic lower bounds against non-interactive locally private protocols:
\begin{theorem}
    \label{thm:expected-max-lb}
    Fix any $\priv>0$. Any non-interactive (public- or private-coin) protocol $\Pi$ for distribution estimation from $\ns$ users must have minmax expected $\lp[\infty]$ error
    \[
        \bigOmega{ \max\Paren{\sqrt{\frac{\log\ab}{\ns(e^\priv-1)^2}},\sqrt{\frac{\log\ab}{\ns e^\priv}}, \frac{\log\ab}{\ns\priv}} }
    \]
\end{theorem}
\begin{proof}
Suppose there exists a (non-interactive, public- or private-coin) $\priv$-LDP protocol $\Pi$ for $\ns$ users which learns any $\p$ to expected $\lp[\infty]$ error $\dst$ when each user gets an independent sample from $\p$ as input:
\begin{equation}
    \bEE{\norminf{ \p-\hat{\p} }} \leq \dst
\end{equation}
where $\hat{\p}$ is the output of $\Pi$ when run on $X_1,\dots, X_\ns \sim \p$.

Now, consider the family $\mathcal{P}_{\dst} = \{\p_z\}_{z\in[\ab]}$ of probability distributions over $[\ab]$, where, for $z\in[\ab]$, $\p_z$ is defined by
\begin{equation}
    \p_z(x) = \frac{1-4\dst}{\ab} + 4\dst\indic{z=x}, \qquad x\in[\ab]
\end{equation}
(that is, $\p_z$ is a mixture of the uniform distribution $\uniform_\ab$ and a point mass on $z$, with mixture coefficients $1-4\dst$ and $4\dst$). Consider the following process: first, we select $Z$ from $[\ab]$ uniformly at random, then generate $\ns$ \iid samples $X_1,\dots, X_\ns$ from $\p_Z$ and run $\Pi$ on these samples, obtaining a hypothesis $\hat{\p}$. We finally set $\hat{Z}$ to be the element of the domain with the largest probability under $\hat{\p}$, \ie
\[
    \hat{Z} = \arg\!\max_{i \in [\ab]} \hat{\p}(i)
\]
(breaking ties arbitrarily). The first claim is that doing so allows one to guess the correct value of $Z$ with high probability: indeed, since the gap between the highest and second highest element of $\p_Z$ is $4\dst$, we have
\begin{equation}
    \label{eq:lb:step1:accuracy}
\probaOf{\hat{Z} = Z} \geq \probaOf{\norminf{ \p-\hat{\p} } <  2\dst} \geq 1 - \frac{\bEE{\norminf{ \p_Z-\hat{\p} }}}{2\dst} \geq \frac{1}{2}
\end{equation}
the second-to-last inequality being Markov's. However, by Fano's inequality applied to the Markov chain $Z - \p_Z - X^\ns - (Y^\ns,R) - \hat{\p} - \hat{Z}$ (where $X^\ns$ is the tuple of \iid samples, and $(Y^\ns,R)$ denotes the tuple of $\ns$ messages, along with the public randomness, resulting from the protocol $\Pi$), we get, recalling that $Z$ is chosen uniformly in $[\ab]$, that
\begin{equation}
    \label{eq:lb:step1:fano}
    \probaOf{\hat{Z} = Z} \leq \frac{\mutualinfo{Z}{(Y^\ns,R)} + \log 2}{\log \ab}
\end{equation}
and so, putting~\cref{eq:lb:step1:accuracy,eq:lb:step1:fano} together, we get
\begin{equation}
    \label{eq:lb:step1}
    \mutualinfo{Z}{(Y^\ns,R)} \geq \frac{1}{2}\log\frac{\ab}{4} = \Omega(\log\ab)\,.
\end{equation}
This gives us the first ingredient: a lower bound on $\mutualinfo{Z}{(Y^\ns,R)}$. For the second, we need to obtain an upper bound on this same mutual information, as a function of $\priv,\ns,\ab,$ and $\dst$. To do so, observe first that, by the chain rule
\[
    \mutualinfo{Z}{(Y^\ns,R)} = \mutualinfo{Z}{Y^\ns \mid R} + \mutualinfo{Z}{R} = \mutualinfo{Z}{Y^\ns \mid R}
\]
the second equality since the public randomness $R$ is independent from $Z$. This is convenient, as the messages $Y^\ns = (Y_1,\dots,Y_\ns)$ are independent conditioned on $R$, and so we get
\begin{equation}
    \label{eq:step2:oneuser}
 \mutualinfo{Z}{(Y^\ns,R)} \leq \sum_{j=1}^\ns \mutualinfo{Z}{Y_j \mid R} \leq \ns\cdot \max_{1\leq j\leq \ns} \mutualinfo{Z}{Y_j \mid R}
\end{equation}
Consider any user $j$, using locally private randomizer $Q = Q_{j,R}\colon [\ab]\to\cY$ (for notational simplicity, we drop afterwards the dependence on $j$ and $R$). Let $\bar{\p} = \bE{Z}{\p_Z}$ denote the average input distribution (over $Z$), \ie the uniform mixture of all $\p_z$'s. Then we can rewrite the mutual information as
\begin{equation}
    \mutualinfo{Z}{Y_j \mid R}
    = \bE{Z}{\kldiv{Q^{\p_Z}}{Q^{\bar{\p}}}}
\end{equation}
where, for a given input distribution $\p$, $Q^{\p}(\cdot ) = \bE{X\sim\p}{Q(\cdot \mid X)}$ denotes the output distribution (over $\cY$) induced by the randomizer $Q$ on input $X\sim\p$. 

\paragraph{First lower bound (good for small $\priv$).} We then proceed by upperbounding the KL divergence by the $\chi^2$ one and unrolling the latter's definition, getting
\begin{align*}
    \mutualinfo{Z}{Y_j \mid R}
    \leq \bE{Z}{\chisquare{Q^{\p_Z}}{Q^{\bar{\p}}}} 
    &= \bE{Z}{\sum_{y\in\cY} \frac{\Paren{\bE{X\sim\p_Z}{Q(y \mid X)}-\bE{X\sim\bar{\p}}{Q(y \mid X)}}^2}{\bE{X\sim\bar{\p}}{Q(y \mid X)}} } \\
    &= \bE{Z}{\sum_{y\in\cY} \frac{\Paren{\sum_{x\in[\ab]} Q(y\mid x) (\p_Z(x) - \bar{\p}(x) )}^2}{\sum_{x\in[\ab]} Q(y\mid x) \bar{\p}(x) } }
\end{align*}
Note, observing that, for our choice of $\mathcal{P}_\alpha$, $\bar{\p}$ is simply the uniform distribution over $[\ab]$, and that for all $x\in[\ab]$ we then have $\p_Z(x) - \bar{\p}(x) = 4\dst\Paren{\indic{Z=x} - \frac{1}{\ab}}$. Then, letting $\Phi(x) = e_x - \frac{1}{\ab}\mathbf{1}_\ab \in \R^\ab$ for $x\in[\ab]$, we get
\begin{align}
    \mutualinfo{Z}{Y_j \mid R}
    &\leq 16\dst^2\ab^2\cdot \bE{Z}{\sum_{y\in\cY} \frac{\Paren{\bE{X\sim\uniform}{Q(y \mid X) \Paren{\indic{Z=X} - \frac{1}{\ab}}}}^2}{\bE{X\sim\uniform}{Q(y \mid X)}} } \notag\\
    &= 16\dst^2\ab^2\cdot \sum_{y\in\cY} \frac{\bE{Z}{\Paren{\bE{X\sim\uniform}{Q(y \mid X) \Paren{\indic{Z=X} - \frac{1}{\ab}}}}^2}}{\bE{X\sim\uniform}{Q(y \mid X)}} \notag\\
    &= 16\dst^2\ab\cdot \sum_{y\in\cY} \frac{\sum_{i=1}^\ab\Paren{\bE{X\sim\uniform}{Q(y \mid X) \Paren{\indic{i=X} - \frac{1}{\ab}}}}^2}{\bE{X\sim\uniform}{Q(y \mid X)}} \notag\\
    &= 16\dst^2\ab\cdot \sum_{y\in\cY} \frac{\sum_{i=1}^\ab\bE{X\sim\uniform}{Q(y \mid X) \Phi(X)_i}^2}{\bE{X\sim\uniform}{Q(y \mid X)}} \label{eq:intermediate:lb:step}
\end{align}
%
%
Note that, for any fixed $1\leq i< j \leq \ab$, we have
\[
\bE{X\sim\uniform}{\Phi_i(X)} = 0,\quad 
\bE{X\sim\uniform}{\Phi_i(X)^2} = \frac{1}{\ab}\Paren{1-\frac{1}{\ab}},\quad
\bE{X\sim\uniform}{\abs{\Phi_i(X)}} = \frac{2}{\ab}\Paren{1-\frac{1}{\ab}},\quad
\]
For simplicity, we write $\bE{X}{\cdot}$ for $\bE{X\sim\uniform}{\cdot}$. We will use the fact that, $Q$ being $\priv$-LDP, we have
\begin{equation}
    \label{eq:guarantee:ldp:smalleps}
    \abs{Q(y \mid x) - \bE{X}{Q(y \mid X)}} \leq (e^\priv-1)\bE{X}{Q(y \mid X)}
\end{equation}
for every $x\in[\ab]$ and $y\in\cY$. With this in hand, starting from~\eqref{eq:intermediate:lb:step}, we can write
\begin{align*}
    \mutualinfo{Z}{Y_j \mid R}
    &\leq 16\dst^2\ab \sum_{y\in\cY} \frac{\sum_{i=1}^\ab\bE{X}{Q(y \mid X) \Phi(X)_i}^2}{\bE{X}{Q(y \mid X)}} \\
    &= 16\dst^2\ab \sum_{y\in\cY} \frac{\sum_{i=1}^\ab\bE{X}{\Paren{Q(y \mid X) - \bE{X}{Q(y \mid X)}} \Phi(X)_i}^2}{\bE{X}{Q(y \mid X)}} \tag{as $\bE{X}{\Phi_i(X)} = 0$} \\
    &\leq 16\dst^2\ab \sum_{y\in\cY} \frac{\sum_{i=1}^\ab\bE{X}{\abs{Q(y \mid X) - \bE{X}{Q(y \mid X)}} \cdot \abs{\Phi(X)_i}}^2}{\bE{X}{Q(y \mid X)}}\\
    &\leq 16\dst^2\ab \sum_{y\in\cY} \frac{\sum_{i=1}^\ab (e^\priv-1)^2\bE{X}{Q(y \mid X)}^2 \bE{X}{\abs{\Phi(X)_i}}^2}{\bE{X}{Q(y \mid X)}} \tag{By~\eqref{eq:guarantee:ldp:smalleps}}\\
    &= 16\dst^2\ab  (e^\priv-1)^2 \sum_{y\in\cY} \sum_{i=1}^\ab \bE{X}{Q(y \mid X)} \bE{X}{\abs{\Phi(X)_i}}^2 \\
    &\leq 64\dst^2  (e^\priv-1)^2 \sum_{y\in\cY} \bE{X}{Q(y \mid X)} \tag{as $\bE{X\sim\uniform}{\abs{\Phi_i(X)}} \leq \frac{2}{\ab}$}\\
    &= 64\dst^2  (e^\priv-1)^2\,.
\end{align*}
Using this last bound along with~\cref{eq:lb:step1,eq:step2:oneuser}, we get that
\begin{equation}
    \frac{1}{2}\log\frac{\ab}{4} \leq \mutualinfo{Z}{(Y^\ns,R) }
    \leq 64\dst^2 \ns  (e^\priv-1)^2\,,
\end{equation}
\ie
\begin{equation}
    \dst \geq \frac{1}{8\sqrt{2}}\sqrt{\frac{\log\frac{\ab}{4}}{\ns  (e^\priv-1)^2}},
\end{equation}
showing the $\bigOmega{\sqrt{\frac{\log\ab}{\ns\priv^2}}}$ lower bound for small $\priv$.
\paragraph{Second lower bound.} 
To get the second term of the lower bound, we resume from~\eqref{eq:intermediate:lb:step}, but this time use the fact that $Q(y \mid X) \leq e^\priv \expect{Q(y \mid X)}$:
\begin{align*}
    \mutualinfo{Z}{Y_j \mid R}
    &\leq 16\dst^2\ab \sum_{y\in\cY} \frac{\sum_{i=1}^\ab\bE{X}{\abs{Q(y \mid X) \Phi(X)_i}}\cdot \bE{X}{\abs{Q(y \mid X) \Phi(X)_i}}}{\bE{X}{Q(y \mid X)}} \tag{By~\eqref{eq:intermediate:lb:step}}\\
    &\leq 16\dst^2\ab \sum_{y\in\cY} \frac{\sum_{i=1}^\ab e^{\priv} \expect{Q(y \mid X)} \bE{X}{|\Phi(X)_i|} \bE{X}{Q(y \mid X) \abs{\Phi(X)_i}} }{\bE{X}{Q(y \mid X)}} \tag{cf. above}\\
    &= 16\dst^2 e^\priv \ab \sum_{y\in\cY} \sum_{i=1}^\ab \bE{X}{|\Phi(X)_i|}\cdot  \bE{X}{Q(y \mid X) \abs{\Phi(X)_i}} \\
    &= 16\dst^2 e^\priv \ab \sum_{i=1}^\ab \bE{X}{|\Phi(X)_i|}\cdot  \bE{X}{\abs{\Phi(X)_i}}  \tag{as $\sum_{y\in\cY} Q(y \mid X) = 1$}\\
    &\leq 64\dst^2 e^\priv \,. \tag{as $\bE{X\sim\uniform}{\abs{\Phi_i(X)}} \leq \frac{2}{\ab}$}
\end{align*}
Combining this again with~\cref{eq:lb:step1,eq:step2:oneuser}, we get 
\begin{equation}
    \dst \geq \frac{1}{8\sqrt{2}}\sqrt{\frac{\log\frac{\ab}{4}}{\ns  e^\priv }},
\end{equation}
showing the second part of the lower bound.

\paragraph{Third lower bound (good for large $\priv$).}
Bounding the $\chi^2$ divergence instead of the KL divergence turns out to be too lossy in that case. We record a simple bound one can obtain by handling directly the latter -- using the same lower bound instance as above. Since, for every $y\in\cY$ and $z\in[\ab]$, $\bE{\p_Z}{Q(y \mid X)}
= (1-4\dst) \bE{\bar{\p}}{Q(y \mid X)} + 4\dst Q(y \mid z)$, we can write
\begin{align*}
    \mutualinfo{Z}{Y_j \mid R}
    &= \bE{Z}{\kldiv{Q^{\p_Z}}{Q^{\bar{\p}}}} \\
    &= \bE{Z}{\sum_{y\in\cY} \bE{\p_Z}{Q(y \mid X)} \log \frac{\bE{\p_Z}{Q(y \mid X)}}{\bE{\bar{\p}}{Q(y \mid X)}} } \\
    &= (1-4\dst)\bE{Z}{\sum_{y\in\cY} \bE{\bar{\p}}{Q(y \mid X)} \log \frac{\bE{\p_Z}{Q(y \mid X)}}{\bE{\bar{\p}}{Q(y \mid X)}} } \\
    &\qquad+ 4\dst \bE{Z}{\sum_{y\in\cY} Q(y \mid Z) \log \frac{\bE{\p_Z}{Q(y \mid X)}}{\bE{\bar{\p}}{Q(y \mid X)}} } \\ %
    &= \underbracket{-(1-4\dst) \bE{Z}{\kldiv{Q^{\bar{\p}}}{Q^{\p_Z}}}}_{\leq 0} + 4\dst \shortexpect_Z{\sum_{y\in\cY} Q(y \mid Z) \log \underbracket{\frac{\bE{\p_Z}{Q(y \mid X)}}{\bE{\bar{\p}}{Q(y \mid X)}}}_{\leq e^\priv} } \\
    &\leq 4\dst\priv \shortexpect_Z{\sum_{y\in\cY} Q(y \mid Z) } = 4\dst\priv\,,
\end{align*}
the last equality as $\sum_{y\in\cY} Q(y \mid z) = 1 $ for every $z$. Again using~\cref{eq:lb:step1,eq:step2:oneuser}, we get 
\begin{equation}
    \dst \geq \frac{\log\frac{\ab}{4}}{8 \ns \priv },
\end{equation}
that is, a lower bound of $\bigOmega{\frac{\log\ab}{\ns\priv}}$, better for large $\priv$.
\end{proof}

%% file: amplification-by-shuffling.tex
As mentioned in the introduction, one of the key motivations for studying locally private histogram estimation in the low privacy regime is the implication for histogram estimation in the \emph{shuffle} model of privacy (see, \eg~\cite{Ghazi0MP20,GhaziG0PV21,BalcerC20,CheuZ22}), in light of the ``plug-and-play'' amplification-by-shuffling results allowing to ``translate'' the former into the latter. Specifically, we will use the following result of Feldman, McMillan, and Talwar:
\begin{theorem}[{\cite[Theorem~3.1]{FeldmanMT21}}]\label{theo:amplification}
 For any domain $\cX$, let $R\colon\cX\to\cY$ be an $\priv_{L}$-DP local randomiser; and let $S$ be the algorithm that given a tuple of $\ns$ messages $\vec{y}\in\cY^\ns$, samples a uniform random permutation $\pi$ over $[\ns]$ and outputs $(\vec{y}_{\pi(1)},\dots,\vec{y}_{\pi(\ns)})$.
 Then for any $\privdelta\in(0,1]$ such that $\priv_{L}\le\log\frac{\ns}{16\log(2/\privdelta)}$, $S\circ R^\ns$ is $(\priv, \privdelta)$-DP,  where
\[
\priv\le \log\Paren{1+8\frac{e^{\priv_{L}}-1}{e^{\priv_{L}}+1}\left(\sqrt{\frac{e^{\priv_{L}}\log(4/\privdelta)}{\ns}}+\frac{e^{\priv_{L}}}{\ns}\right) }\,.
\]
In particular, if $\priv_{L}\geq 1$ then $\priv = O\Paren{\sqrt{e^{\priv_{L}}\log(1/\privdelta)/\ns}}$, and if $\priv_{L} < 1$ then $\priv = O\Paren{\priv_{L}\sqrt{\log(1/\privdelta)/\ns}}$.
\end{theorem}

This implies the following:
\begin{lemma}[Amplification by shuffling]
\label{lemma:amp:shuffling}
Fix any $\privdelta\in(0,1]$, $\priv\in(0,1]$, and $\ns$ such that
$
\priv > 16\sqrt{\log(4/\privdelta)/\ns}
$. 
Then, for \[
\priv_L \eqdef \log\frac{\priv^2\ns}{256\log(4/\privdelta)} = \bigTheta{\log\frac{\priv^2\ns}{\log(1/\privdelta)}},
\]
shuffling the messages of $\ns$ users using the same $\priv_L$-LDP randomizer satisfies (robust) $(\priv,\privdelta)$-shuffle differential privacy.
\end{lemma}
\begin{proof}
Note that for $\priv,\privdelta$ as in the statement and $\priv_L$ as defined, we have 
$0 < \priv_{L}\le\log\frac{\ns}{16\log(2/\privdelta)}$. Applying~\cref{theo:amplification}, we get $(\priv',\privdelta)$ privacy for
\[
\priv' = \log\Big(1+8\underbrace{\frac{e^{\priv_{L}}-1}{e^{\priv_{L}}+1}}_{\leq 1}\Big(\frac{\priv}{16}+\underbrace{\frac{\priv^2}{256\log(4/\privdelta)}}_{\leq \priv/16}\Big) \Big)\leq \priv
\]
which proves the statement.
\end{proof}
\begin{theorem}
    \label{theo:pgr:shuffle}
    For $\ns = \bigOmega{\frac{\log(1/\privdelta)}{\priv^2}}$ and $\priv \in(0,1]$, Shuffled Projective Geometry Response achieves maximum error $\bigO{\frac{\sqrt{\log(\ab)\log(1/\privdelta)}}{n\priv}}$, with $O(\log\ab)$ bits of communication (and one single message) per user.
\end{theorem}
\begin{proof}
For $\ns \geq \frac{500}{\priv^2}\log\frac{4}{\privdelta}$, the restriction on $\priv$ from~\cref{lemma:amp:shuffling} holds, and, setting $\priv_L$ as in the lemma, we also have $\priv_L = \Omega(1)$. We then invoke the bound of~\cref{theo:optimal:pgr:ub}, focusing on first part of the bound, which dominates when $\frac{\log\ab}{n}\cdot \log^2 n\leq\frac{\priv_L^2}{e^{\priv_L}}$. This leads to an upper bound on the error of the order%
    \begin{align*}
        \sqrt{\frac{\log\ab}{\ns e^{\priv_L}}} &= \sqrt{\frac{\log(\ab)}{\ns\priv^2\frac{\ns}{\log(1/\privdelta)}}}
            =\frac{\sqrt{\log(\ab)\log(1/\privdelta)}}{\ns\priv}
    \end{align*}
    as desired. It only remains to argue that the first term of the bound did, indeed, dominate the error. As mentioned above, this is the case whenever 
    \[
    \frac{e^{\priv_L}}{\priv_L^2}\cdot \log \ab \ll \frac{n}{\log^2 n}
    \]
    for which a weaker, sufficient condition is $\frac{e^{\priv_L}}{\priv_L^2} \ll \frac{n}{\log^2 n}$, that is, $n \gg e^{\priv_L}$. But this follows from our setting of $\priv_L$, such that $e^{\priv_L} = \underbrace{\tfrac{\priv^2}{256\log(4/\privdelta)}}_{\ll 1} \cdot n$.
\end{proof}
Comparing this result with the summary of local, shuffle, and central histogram error bounds available in~\cite[Table~1]{CheuZ22} shows that with shuffled PGR achieves the best error of any protocol which sends a constant number of messages.

More specifically, focusing on 3 representative known protocols: the only known one-message-per-user protocol, due to~\cite{CheuZ22}, achieves much worse error,\footnote{And has the same restriction $\ns = \Omega(\log(1/\privdelta)/\priv^2)$ on the parameters.} and requires $\ab$ bits of communication per user; while the protocol of~\cite{GhaziG0PV21}, which achieves the same error as~\cref{theo:pgr:shuffle}, uses $\ab^{\Omega(1)}$ $(\log\ab)$-bit messages per user. Finally, a protocol of~\cite{balcer2020} does achieve better error, but at the cost of performing $\ab+1$ rounds, each with $\log\ab$ communication per user. Thus, our bounds demonstrates that shuffled PGR achieve state-of-the-art $\lp[\infty]$--error with only \emph{one} message of $\log\ab$ bits.

%% file: concentration-figures.tex
\begin{figure}[h]
    \centering
    \includegraphics[width=0.99\textwidth]{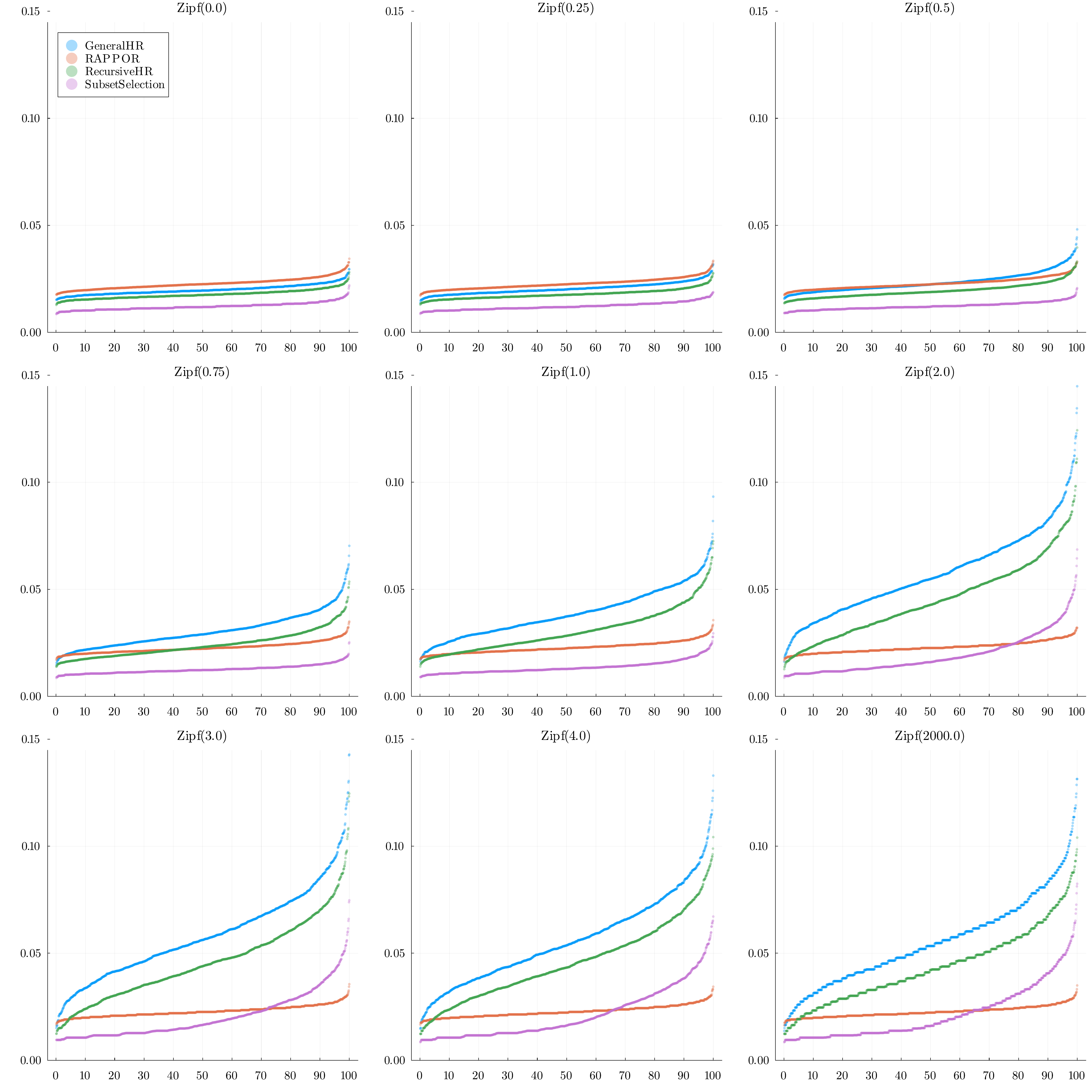}
    \caption{$\lp[\infty]$ error in $x\%$ of runs, with 1000 repeats per protocol. Distributions are $\operatorname{Zipf}(\alpha)$, where $p_i\propto i^{-\alpha}$, larger values of $\alpha$ give more concentrated distributions, $\alpha=2000$ has its entire mass on a single point, while $\alpha=0$ is the uniform distribution. All experiments with $\priv=5$, $\ab=500$, $\ns=1000$.}
    \label{fig:err-by-concentration}
\end{figure}

\Cref{fig:err-by-concentration} demonstrates a relationship between the concentration of a distribution and the error of the protocol. The fact that these protocols perform best on a uniform distribution should not be surprising, given that the randomiser has the effect of flattening the input distribution. What is most interesting is that RAPPOR maintains the same error distribution invariant of input distribution, and that the subset-based protocols do not. We propose that this is due to the independence of coordinates under RAPPOR.